\documentclass[12pt, letterpaper]{article}

\usepackage[margin=2.5cm]{geometry} 
\usepackage[T1]{fontenc}           
\usepackage{lmodern}               
\usepackage{setspace}              
\usepackage{authblk}               

\usepackage{amsmath}               
\usepackage{amssymb}               
\usepackage{amsthm}                
\usepackage{bm}                    

\newtheorem{proposition}{Proposition}

\usepackage{graphicx}              
\usepackage[ruled, lined, linesnumbered]{algorithm2e} 

\usepackage{caption}
\makeatletter
\renewcommand{\fnum@figure}{\textbf{Figure \thefigure}}
\renewcommand{\fnum@table}{\textbf{Table \thetable}}
\makeatother

\usepackage[authoryear]{natbib}     
\usepackage{url}                   

\usepackage{xcolor}
\definecolor{mycustomblue}{HTML}{023C85}

\usepackage{hyperref}
\hypersetup{
    colorlinks=true,
    linkcolor=mycustomblue,        
    citecolor=mycustomblue,        
    urlcolor=mycustomblue,         
    filecolor=mycustomblue         
}

\Huge\title{A Probabilistic Framework for Estimating the Modal Age at Death}
\author{%
  \begin{minipage}[t]{0.4\textwidth}
    \centering
    Silvio C. Patricio\thanks{\texttt{silca [at] sam.sdu.dk}; ORCID: 0000-0002-1025-5571} \\
  \end{minipage}%
  \hspace{-0.05\textwidth}%
  \begin{minipage}[t]{0.4\textwidth}
    \centering
    Paola Vazquez-Castillo\thanks{\texttt{pavaz [at] sdu.dk}; ORCID: 0000-0001-9680-4721}\\
  \end{minipage}\\
  \small{Interdisciplinary Center on Population Dynamics\\ University of Southern Denmark}\\
}

\date{}

\begin{document}
\maketitle


\begin{abstract}
    \noindent The modal age at death is used to track longevity, but existing methods report it as a single number. We estimate its posterior distribution. Death counts are modeled as Poisson given the death rates and exposures, and posterior draws of the rates, passed through the life table, give the probability that each age is the mode and a credible interval for an interpolated modal age. With independent gamma priors, the posterior is computed exactly and concentrates on the true mode as the population grows. Because single-year rates are noisy, we also use a second-order random-walk prior on the log rates, which shares information across neighboring ages and estimates the degree of smoothing from the data; this posterior is computed with a Laplace approximation, which we check against posterior sampling. In simulations, the random-walk posterior mean was 1.4 to 10.6 times more accurate than the posterior mean under independent priors. When the distribution of deaths was unimodal, its errors were close to those of P-spline and Gompertz estimates, which were slightly more accurate; with a secondary peak, P-splines were the most accurate. The random-walk credible intervals were conservative, covering the true mode in 98\% to 100\% of replications, whereas bootstrap intervals for the other estimators either covered it less often or were much wider. For women in Denmark, France, Japan and Sweden since 1960, the random-walk estimates hardly changed with the starting age of the estimation, and the posterior gives the probability that the modal age declined in 2020.
\end{abstract}

\section{Introduction}

The modal age at death, the age at which the life-table distribution of deaths peaks, is a standard indicator of longevity. It is determined by old-age mortality and, unlike life expectancy, does not depend on mortality at young ages. Its changes therefore track the postponement of deaths into later life \citep{kannisto2001mode, canudas2008modal, horiuchi2013modal}. Its trends have been used to separate the shifting of the distribution of deaths from its compression \citep{canudas2008modal, bergeron2015decomposing}.

The mode is harder to estimate than a mean. The distribution of deaths is flat near its peak, so small random fluctuations in single-year death counts can move the highest point by a year or more. Three types of methods address this. Parametric models such as the Gompertz give the mode in closed form from the fitted parameters \citep{canudas2008modal, missov2015gompertz}. Smoothing methods fit a P-spline to the death counts and take the maximum of the smoothed distribution \citep{ouellette2011changes, horiuchi2013modal}. A third group uses the life-table deaths directly. \citet{kannisto2001mode} interpolates between the ages around the highest count. The discretized derivative tests (DDT) of \citet{vazquez2024longevity} locate the age at which discrete versions of the conditions defining a mode are met.

These methods are usually presented as point estimators, with uncertainty assessed separately by the bootstrap or the delta method \citep{andreev2010spreadsheet, keyfitz2005applied}. The resulting confidence intervals describe how much an estimate would vary over hypothetical repetitions of the data; they do not assign a probability to each candidate age. A point estimate does not show when two or three neighboring ages have almost equal support in the data.

We estimate the posterior distribution of the modal age by passing uncertainty in the death rates through the life table. The approach builds on \citet{andreev2015confidence}, who obtained uncertainty limits for life expectancy from posterior rate distributions, and on Bayesian P-splines, which account for uncertainty in the smoothed schedule and its smoothing parameter \citep{lang2004bayesian}. We apply these ideas to the modal age, which jumps between ages when two values of the life-table deaths cross. We consider independent gamma priors, which give an exact posterior and some closed-form results, and a random-walk prior on the log rates, which shares information across neighboring ages.

We compare the two versions with the Gompertz, P-spline and DDT estimators in a simulation study with a known modal age and in data for women in four countries. The comparison concerns modal-age estimation and its uncertainty, not Bayesian versus frequentist inference. We assess point-estimate accuracy, interval coverage and width, and sensitivity to the age range used. Section~\ref{sec:method} develops the method, Section~\ref{sec:simulation} reports the simulation study, Section~\ref{sec:application} the application, and Section~\ref{sec:conclusion} discusses when the method should be used.

\section{A Posterior Distribution for the Modal Age at Death}
\label{sec:method}

A distribution of modal ages describes uncertainty about the location of the adult peak in life-table deaths. It is distinct from the life-table distribution itself, which describes variation in individual lifespans under a given mortality schedule. Separating rate estimation from the life-table transformation lets us compare different prior assumptions while retaining the same demographic quantity.

\subsection{The modal age is a function of the death rates}
\label{sec:target}

Let $x = x_0, x_0+1, \dots, \omega-1$ index single-year age intervals $[x, x+1)$, starting at an age $x_0$ and ending before the open-ended interval $\omega+$. For each age we observe the number of deaths $D_x$ and the person-years of exposure $E_x$. Let $m_x$ be the unknown death rate at age $x$, and write $\bm{m} = (m_{x_0}, \dots, m_{\omega-1})^\top$, $\bm{D}$ and $\bm{E}$ for the corresponding vectors. There are $K = \omega - x_0$ ages.

Assuming a constant force of mortality within each interval, the life table is
\begin{equation}
q_x = 1 - e^{-m_x}, \qquad
\ell_x = \exp\Big(-\sum_{y=x_0}^{x-1} m_y\Big), \qquad
d_x = \ell_x \, q_x ,
\label{eq:lifetable}
\end{equation}
with $\ell_{x_0} = 1$ \citep{preston2000demography}, so that $d_x$ is the probability that a person alive at age $x_0$ dies in the interval $[x, x+1)$. The modal age at death is the age with the largest $d_x$,
\begin{equation}
M(\bm{m}) = \operatorname*{arg\,max}_{x_0 \le x < \omega} \; d_x(\bm{m}).
\label{eq:mode_def}
\end{equation}

Mortality below $x_0$ multiplies every $d_x$ by the same factor, the probability of surviving from birth to $x_0$, and does not move the maximum. The starting age $x_0$ is therefore needed only to exclude the peak of deaths in infancy, and the mode depends on the rates at ages $x_0$ and above alone. The target is also the \emph{life-table} distribution of deaths, not the observed distribution $D_x / \sum_y D_y$, which reflects the age structure of the population as well as its mortality.

If $\bm{m}$ were known, $M$ would be a fixed number. The uncertainty concerns our knowledge of $\bm{m}$, which is estimated from a finite number of deaths. In a Bayesian analysis this uncertainty is described by a probability distribution for the rates given the observed data, the posterior distribution $p(\bm{m} \mid \bm{D})$ (exposures are treated as known and left implicit). Each rate schedule has its own modal age. The probability that age $x$ is the mode is the posterior probability assigned to schedules whose mode is $x$:
\begin{equation}
\mathbb{P}(M = x \mid \bm{D}) = \int \mathbb{I}\{M(\bm{m}) = x\} \, p(\bm{m} \mid \bm{D}) \, d\bm{m},
\label{eq:post_mode}
\end{equation}
where $\mathbb{I}\{\cdot\}$ is the indicator function. Equation~\eqref{eq:post_mode} is the central object of the paper. The next subsections specify $p(\bm{m} \mid \bm{D})$ and show how to evaluate the integral.

\subsection{Death counts and the posterior distribution of the rates}
\label{sec:posterior}

Following \citet{brillinger1986biometrics}, we model the death count at each age as
\begin{equation}
D_x \mid m_x \sim \text{Poisson}(m_x E_x),
\label{eq:poisson}
\end{equation}
independently across ages. This is the standard likelihood behind the occurrence--exposure rate $D_x / E_x$, which is its maximum likelihood estimator. We do not condition on the total number of deaths: conditioning independent Poisson counts on their sum gives a multinomial distribution whose cell probabilities are the expected shares of \emph{observed} deaths \citep{bishop1975discrete}, which mix mortality with age structure and are not the life-table $d_x$.

For each rate we take a gamma prior, $m_x \sim \text{Gamma}(\alpha, \beta)$, independently across ages. The gamma family is conjugate to the Poisson likelihood \citep{gelman1995bayesian}, so the posterior is
\begin{equation}
m_x \mid \bm{D} \sim \text{Gamma}(\alpha + D_x, \; \beta + E_x),
\label{eq:gamma_post}
\end{equation}
independently across ages, with shape $\alpha + D_x$ and rate $\beta + E_x$. We use $\alpha = 1/2$ and $\beta = 0$, which corresponds to the improper Jeffreys prior $p(m_x) \propto m_x^{-1/2}$ for a Poisson rate \citep{gelman1995bayesian}; the posterior \eqref{eq:gamma_post} is proper whenever $E_x > 0$, including ages with no deaths. The posterior mean of $m_x$ is then $(D_x + 1/2)/E_x$, essentially the observed rate, and its coefficient of variation is $1/\sqrt{D_x + 1/2}$: about 3\% with 1{,}000 deaths at that age and 1\% with 10{,}000. The prior matters only at ages with very few deaths.

The posterior rates are independent. The life-table deaths are dependent: through $\ell_x$, each $d_x$ depends on every rate below age $x$. Draws of $\bm{m}$ passed through \eqref{eq:lifetable} preserve this dependence.

\subsection{From the posterior of the rates to the distribution of the mode}
\label{sec:mode_dist}

Because the posterior \eqref{eq:gamma_post} is a product of known gamma distributions, \eqref{eq:post_mode} can be evaluated by Monte Carlo integration with exact draws. For $s = 1, \dots, S$, draw $m_x^{(s)}$ from \eqref{eq:gamma_post} at every age, compute $\bm{d}^{(s)}$ from \eqref{eq:lifetable}, and record $M^{(s)} = M(\bm{m}^{(s)})$. Then
\begin{equation}
\widehat{\mathbb{P}}(M = x \mid \bm{D}) = \frac{1}{S} \sum_{s=1}^{S} \mathbb{I}\{M^{(s)} = x\}.
\label{eq:mc}
\end{equation}
The standard error of each estimated probability is at most $1/(2\sqrt{S})$, or $0.005$ with $S = 10^4$ draws.

We summarize the posterior by the most probable age and its probability, or by a credible set: the smallest set of ages whose probabilities add up to at least $0.95$. For point estimates and intervals on a continuous age scale, we interpolate within each draw. We fit a parabola through the life-table deaths of the modal interval $[x, x+1)$, where $x = M^{(s)}$, and its two neighbors. Each value is placed at its interval midpoint ($x - 1/2$, $x + 1/2$ and $x + 3/2$), giving the vertex
\begin{equation}
\tilde{M}^{(s)} = x + \frac{d_x^{(s)} - d_{x-1}^{(s)}}{\big(d_x^{(s)} - d_{x-1}^{(s)}\big) + \big(d_x^{(s)} - d_{x+1}^{(s)}\big)},
\label{eq:interp}
\end{equation}
which is the interpolation of \citet{kannisto2001mode} applied to each draw. The posterior mean, median and mode of $\tilde{M}$ are point estimates of the modal age, and its 2.5\% and 97.5\% quantiles give a 95\% credible interval. The modal interval $M$ in \eqref{eq:mode_def} and its interpolation $\tilde{M}$ are distinct from the mode of an underlying continuous lifespan density, which the interpolation approximates. In noise-free Gompertz schedules, the interpolation error is $-0.013$ to $-0.029$ years (Appendix~\ref{app:implementation}), small relative to the uncertainty in national data.

\subsection{Local modes}
\label{sec:analytic}

The structure of the life table also gives an exact result for \emph{local} modes. Age $x$ is a local mode if $d_{x-1} < d_x > d_{x+1}$. The factors common to the three ages cancel in \eqref{eq:lifetable}. Each comparison then involves only two neighboring rates, which are independent under \eqref{eq:gamma_post}. The posterior probability that $x$ is a local mode therefore reduces to a one-dimensional integral (Proposition~\ref{prop:local} in Appendix~\ref{app:analytic}). For small rates, age $x$ is a local mode when
\begin{equation*}
m_x - m_{x-1} \gtrsim m_x^2 \gtrsim m_{x+1} - m_x ,
\end{equation*}
the discrete form of the condition $\mu'(M) = \mu(M)^2$ that defines the mode of a continuous lifespan density \citep{canudas2008modal} and that underlies one of the DDT conditions. Summing the local-mode probabilities over ages gives the posterior expected number of local modes, a measure of how rough the single-year schedule is. The global mode has no comparable reduction, because $d_x$ depends on all rates below $x$; this is why we evaluate \eqref{eq:post_mode} by Monte Carlo.

\subsection{Concentration in large populations}
\label{sec:convergence}

The following proposition states the conditions under which the posterior concentrates on one modal age.

\begin{proposition}
\label{prop:convergence}
Let $\bm{m}^0$ be the true rates, and suppose that $\bm{d}(\bm{m}^0)$ has a unique maximizer $x^\star$, with gap $\delta = d_{x^\star}(\bm{m}^0) - \max_{j \neq x^\star} d_j(\bm{m}^0) > 0$. Let the data follow \eqref{eq:poisson} with $\bm{m} = \bm{m}^0$ and exposures $E_x = c_x N$, $c_x > 0$. Then, as $N \to \infty$,
\[
\mathbb{P}(M = x^\star \mid \bm{D}) \xrightarrow{\;p\;} 1 .
\]
\end{proposition}

The proof is in Appendix~\ref{app:proof}. The speed of concentration depends on the gap $\delta$ relative to the posterior standard deviation of $d_{x^\star} - d_j$, which decreases roughly with the inverse square root of the number of deaths. Because the distribution of deaths is flat near its peak, $\delta$ is often small, and the posterior can remain spread over two or three ages even in large populations. This spread is information that a point estimate does not convey: the data do not distinguish between these ages.

\subsection{Posterior and bootstrap probabilities}
\label{sec:why_posterior}

A common alternative is the bootstrap: replicate counts are drawn from the estimated rates, $D_x^\ast \sim \text{Poisson}(\hat m_x E_x)$, the mode is recomputed in each replicate, and the frequency of each age is reported (Appendix~\ref{app:bootstrap}). The computations average over different quantities. The bootstrap holds the rates at their estimates and varies the data, so its frequencies describe the sampling variability of an estimator. The posterior holds the data at what was observed and varies the rates, so it gives the probability that each age is the modal age, given the data, the Poisson model and the prior. With many deaths and a clear gap between the most likely age and its neighbors, both concentrate on the same age. They can differ when neighboring ages are nearly tied. Neither construction guarantees intervals with nominal coverage in finite samples.

When the true mode is known, both kinds of interval can be assessed by their coverage and width. Section~\ref{sec:simulation} uses these criteria to compare the credible intervals of the posterior with bootstrap intervals for the other estimators.

\subsection{Borrowing strength across ages: a random-walk prior}
\label{sec:rw2}

Under independent priors, each rate is estimated using only the deaths at its own age. Random fluctuations in single-year counts can produce several local maxima near the peak and spread the posterior across neighboring ages. At adult ages, the logarithm of the death rate increases almost linearly with age \citep{gompertz1825nature, missov2015gompertz}. Graduation methods exploit this regularity; encoding it in a prior lets neighboring ages share information.

Let $\eta_x = \log m_x$. We replace the independent gamma priors with a second-order random walk (RW2) on the log rates,
\begin{equation}
\eta_x - 2\eta_{x-1} + \eta_{x-2} \sim \mathcal{N}(0, \lambda^{-1}), \qquad x = x_0 + 2, \dots, \omega - 1,
\label{eq:rw2}
\end{equation}
independently across ages, and keep the Poisson likelihood \eqref{eq:poisson}. The joint prior density is $p(\bm{\eta} \mid \lambda) \propto \lambda^{(K-2)/2} \exp(-\tfrac{\lambda}{2}\, \bm{\eta}^\top \bm{P} \bm{\eta})$, where $\bm{P} = \bm{\Delta}^\top \bm{\Delta}$ and $\bm{\Delta}$ is the $(K-2) \times K$ matrix of second differences \citep{rue2005gaussian}. The prior favors small changes in the rate of increase of mortality from one age to the next. It does not penalize the level or the slope of $\bm{\eta}$: every Gompertz schedule, $\eta_x = \log a + bx$, has zero penalty. We give these two directions a flat prior. The posterior is then proper whenever deaths are observed at two or more ages. The precision $\lambda$ sets how strongly the log rates are pulled toward a straight line. As $\lambda \to \infty$ the log rates are restricted to a straight line, a Gompertz schedule whose two parameters remain uncertain; as $\lambda \to 0$ each rate is informed by its own age alone (with a flat prior on $\log m_x$ rather than the Jeffreys prior). The prior spans the range between a parametric Gompertz model and the unsmoothed schedule; the posterior weights the smoothing strengths according to the likelihood and the prior.

For fixed $\lambda$, the posterior mode of $\bm{\eta}$ maximizes
\begin{equation}
\ell(\bm{\eta}) - \frac{\lambda}{2} \sum_{x = x_0 + 2}^{\omega - 1} \left(\eta_x - 2\eta_{x-1} + \eta_{x-2}\right)^2,
\qquad
\ell(\bm{\eta}) = \sum_{x} \left( D_x \eta_x - E_x e^{\eta_x} \right),
\label{eq:penalized}
\end{equation}
which is the graduation criterion of \citet{whittaker1922new} written for Poisson counts. The second-order penalty also underlies P-splines \citep{eilers1996flexible, camarda2012mortalitysmooth}; the RW2 prior gives its Bayesian form \citep{lang2004bayesian, rue2005gaussian}. Computing the mode for each posterior schedule, while integrating over the smoothing strength $\lambda$, accounts for uncertainty in both the rates and their smoothness.

Both priors target the modal age of the unknown rate schedule; RW2 adds the smoothness assumption in \eqref{eq:rw2}. This smoothness assumption can reduce accuracy where the true schedule bends sharply, for example where the age-at-death distributions of two subpopulations overlap. Section~\ref{sec:simulation} examines this case.

\paragraph{Computation.} The RW2 posterior has no closed form. We approximate it in two layers, following the approach to integrated nested Laplace approximation of \citet{rue2009approximate}. For fixed $\lambda$, the posterior mode $\hat{\bm{\eta}}_\lambda$ is found by penalized Fisher scoring, and the posterior of $\bm{\eta}$ is approximated by $\mathcal{N}(\hat{\bm{\eta}}_\lambda, \bm{H}_\lambda^{-1})$, where $\bm{H}_\lambda$ is the negative Hessian of the log posterior at the mode \citep{tierney1986accurate}. For $\lambda$, we place a flat prior on $\log_{10} \lambda$ over an equally spaced grid and weight each grid value by the Laplace approximation of its marginal likelihood. Each posterior draw takes a value of $\lambda$ from these weights, then $\bm{\eta}$ from the corresponding Gaussian, and passes $\bm{m} = \exp(\bm{\eta})$ through the life table. Algorithm~\ref{alg:estimation} summarizes both versions of the method, which differ only in how the rates are drawn. Appendix~\ref{app:computation} gives the details and compares the approximation with posterior sampling by Markov chain Monte Carlo, in national data and in small simulated populations.

\begin{algorithm}[htb!]
\caption{Posterior distribution of the modal age at death}
\label{alg:estimation}
\small{
  \SetAlgoLined
  \KwIn{Deaths $D_x$ and exposures $E_x$ for $x_0 \le x < \omega$; number of draws $S$; for RW2, a grid $\lambda_1, \dots, \lambda_J$}
  \KwOut{$\widehat{\mathbb{P}}(M = x \mid \bm{D})$ for each age $x$; draws $\tilde{M}^{(1)}, \dots, \tilde{M}^{(S)}$}
  \If{RW2 prior}{
    \For{$j = 1, \dots, J$}{
      Find $\hat{\bm{\eta}}_{\lambda_j}$ and $\bm{H}_{\lambda_j}$ by iterating \eqref{eq:scoring}; compute the weight $w_j \propto p(\bm{D} \mid \lambda_j)$ from \eqref{eq:marglik} \;
    }
  }
  \For{$s = 1, \dots, S$}{
    \eIf{independent gamma priors}{
      Draw $m_x^{(s)} \sim \text{Gamma}(\alpha + D_x, \, \beta + E_x)$ for every age \;
    }{
      Draw $j$ with probabilities $w_1, \dots, w_J$, draw $\bm{\eta}^{(s)} \sim \mathcal{N}(\hat{\bm{\eta}}_{\lambda_j}, \bm{H}_{\lambda_j}^{-1})$ and set $\bm{m}^{(s)} = \exp(\bm{\eta}^{(s)})$ \;
    }
    Compute $\bm{d}^{(s)}$ from \eqref{eq:lifetable}, $M^{(s)} = \operatorname{arg\,max}_x d_x^{(s)}$ and $\tilde{M}^{(s)}$ from \eqref{eq:interp} \;
  }
  $\widehat{\mathbb{P}}(M = x \mid \bm{D}) \leftarrow S^{-1} \sum_{s} \mathbb{I}\{M^{(s)} = x\}$ for each age $x$ \;
}
\end{algorithm}

Because the RW2 prior links neighboring ages, the fitted schedule depends on the age range used: the rates near $x_0$ influence the fit at older ages. Under independent priors, $x_0$ affects the posterior of $M$ only through the ages it excludes. Section~\ref{sec:application} examines how much this matters in practice. The exact local-mode result of Section~\ref{sec:analytic} also does not carry over, because the rates are dependent a posteriori; the Monte Carlo estimate \eqref{eq:mc} does.

\section{Simulation Study}
\label{sec:simulation}

\paragraph{Design.} Ages at death follow a mixture of two Gompertz distributions with the same slope $b = 0.14$,
\begin{equation}
f(x) = w\, g_1(x) + (1 - w)\, g_2(x), \qquad
g_k(x) = a_k e^{b x} \exp\!\Big(-\frac{a_k}{b}\big(e^{b x} - 1\big)\Big), \quad a_k = b\, e^{-b M_k},
\label{eq:sim_mixture}
\end{equation}
where $g_1$ and $g_2$ have modal ages $M_1 = 90$ and $M_2 = 65$. The mixture represents two population groups with different mortality levels. We consider four scenarios, $w = 1$, $0.8$, $0.7$ and $0.6$, in which $0\%$, $20\%$, $30\%$ and $40\%$ of deaths occur in the earlier group. With $w < 1$ the density has a secondary peak between ages 66 and 68, whose height is $35\%$, $52\%$ and $75\%$ of the main peak. The true modal age is 90 in every scenario. For a radix $N$, exposures are the expected person-years $E_x = N \int_x^{x+1} S(t)\, dt$, with $S$ the survival function of \eqref{eq:sim_mixture}, and death counts are drawn as $D_x \sim \text{Poisson}(m_x E_x)$, with $m_x$ the true central death rate, for ages $0, \dots, 109$ and $110+$. We use $N = 10^3, 10^4, 10^5$ and $10^6$, from the annual deaths of one sex in a small region to those of a large country, and replicate each of the 16 combinations of scenario and population size 500 times. All methods use the starting age $x_0 = 40$.

\paragraph{Methods compared.} We compute the posterior with independent gamma priors (``raw posterior'') and the posterior with the RW2 prior, each with $S = 5{,}000$ draws, and summarize each by the posterior mean, median and mode of $\tilde{M}$ and the 95\% credible interval. We compare these summaries with three point estimators. For the Gompertz estimator, $\mu(x) = a e^{bx}$ is fitted by Poisson maximum likelihood to ages $x_0$ and above, and the mode is $\log(b/a)/b$ \citep{canudas2008modal, missov2015gompertz}. For the P-spline estimator, the death counts are smoothed with a Poisson P-spline, with the smoothing parameter chosen by the Bayesian information criterion as in \citet{camarda2012mortalitysmooth}, and the mode is the maximum of the fitted lifespan density \citep{ouellette2011changes}. We use a modified implementation of the DDT method of \citet{vazquez2024longevity}, adding half a year to align its age convention with the other estimators and specifying a rule for matching multiple zero crossings. Appendix~\ref{app:implementation} explains both choices. For the three point estimators, 95\% confidence intervals are obtained by a plug-in Poisson bootstrap: replicate counts $D^\ast_x \sim \text{Poisson}(D_x)$ with fixed exposures, 200 replicates, and the 2.5\% and 97.5\% quantiles of the re-estimated modes. For the P-spline, the smoothing parameter is held at the value selected on the original data. Accuracy is measured by the bias and the mean absolute error with respect to the true mode, and intervals by their coverage and mean width. With 500 replications, the Monte Carlo standard error of a coverage near $0.95$ is about $0.01$ \citep{morris2019using}.

\begin{figure*}[htb!]
    \centering
    \includegraphics[width=\linewidth]{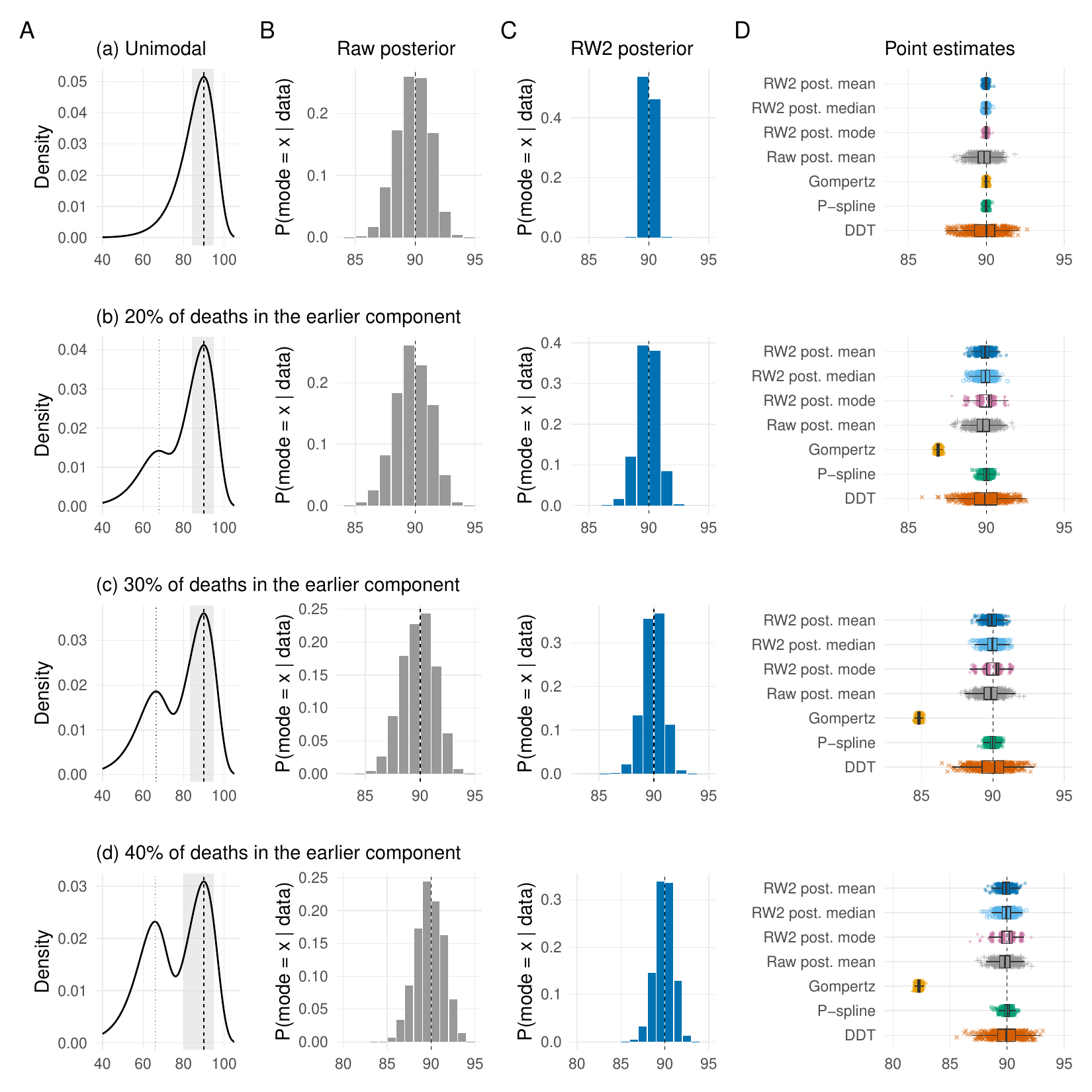}
    \caption{Simulation study with $N = 10^4$. Each row is a scenario. (A) True density of the age at death; the dashed line marks the true mode, the dotted line the secondary peak, and the shaded area the ages shown in B to D. (B) Posterior distribution of the modal age with independent gamma priors and (C) with the RW2 prior, averaged over 500 replications. (D) Point estimates across replications.}
    \label{fig:mc_scenarios}
\end{figure*}

\begin{figure}[htb!]
    \centering
    \includegraphics[width=\linewidth]{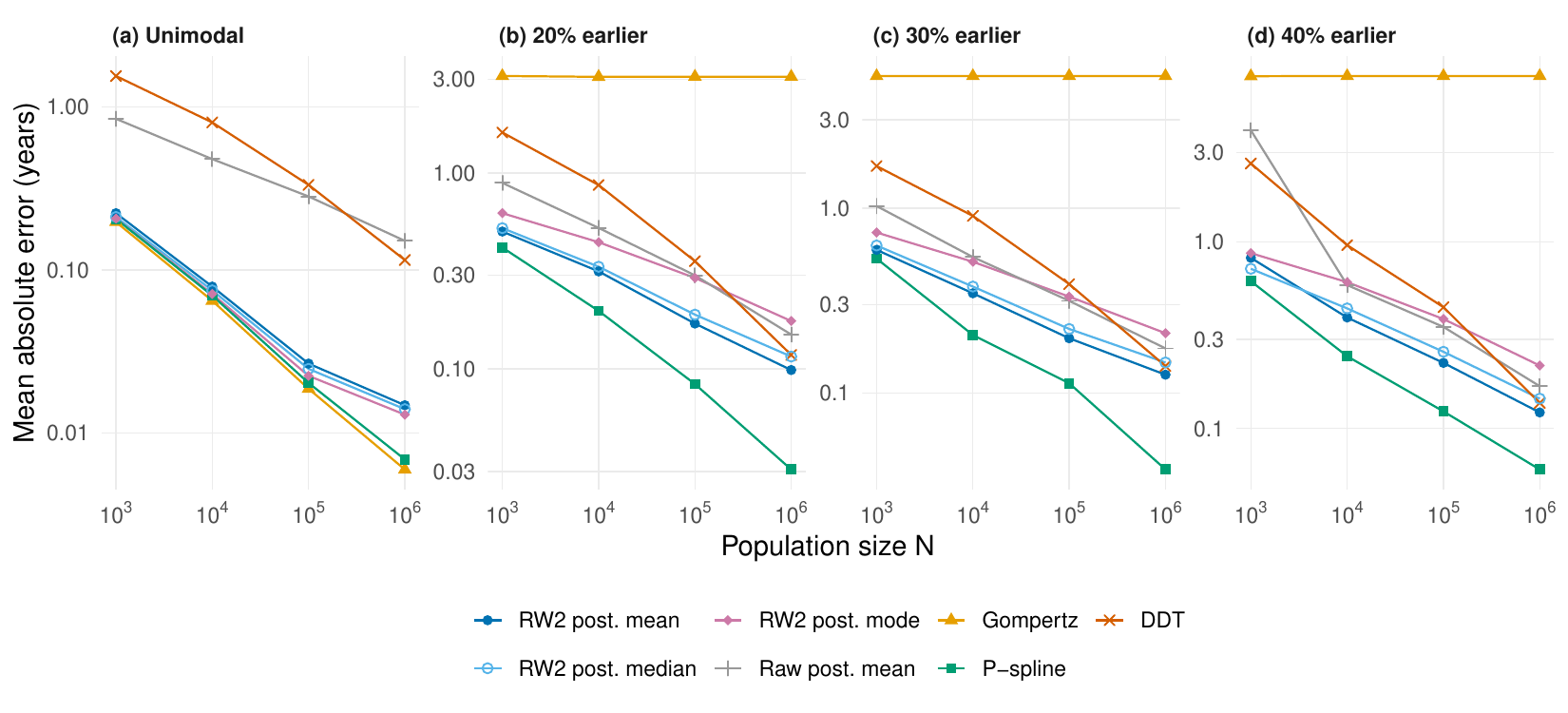}
    \caption{Mean absolute error of the estimated modal age by population size and scenario (logarithmic scales), 500 replications.}
    \label{fig:mc_mae}
\end{figure}

\paragraph{Choice of posterior summary.} For the RW2 posterior, the mean, median and mode of $\tilde{M}$ are about equally accurate in the unimodal scenario (mean absolute errors of $0.08$, $0.07$ and $0.07$ years at $N = 10^4$; Figure~\ref{fig:mc_mae}). With a secondary peak the posterior mean is the most accurate, with errors at $N = 10^4$ of $0.31$, $0.35$ and $0.39$ years in scenarios (b) to (d), against $0.33$ to $0.44$ for the median and $0.44$ to $0.61$ for the mode. The exception is scenario (d) with $N = 10^3$. Posterior mass near the secondary peak lowers the mean (bias $-0.35$), and the median is more accurate ($0.72$ against $0.82$). The raw posterior shows the same ordering. We therefore use the posterior mean as the point estimate.

\paragraph{RW2 against independent priors.} Sharing information across ages concentrates the posterior of the mode (Figure~\ref{fig:mc_scenarios}, panels B and C). In the unimodal scenario, the mean absolute error of the RW2 posterior mean is lower than that of the raw posterior mean by a factor of $3.8$ to $10.6$, for example $0.08$ against $0.48$ years at $N = 10^4$. With a secondary peak, the reduction is a factor of $1.4$ to $1.8$, with a larger reduction in scenario (d) at $N = 10^3$ ($0.82$ against $3.96$ years). The raw posterior is biased downward in small populations ($-0.48$ years in the unimodal scenario at $N = 10^3$), whereas the bias of the RW2 posterior mean stays within $0.10$ years of zero except in scenario (d) at $N = 10^3$.

\paragraph{Comparison of point estimates.} When the distribution of deaths is unimodal, the RW2 posterior mean, the P-spline and the Gompertz estimators all have small errors, although those of the RW2 posterior mean are larger: mean absolute errors of $0.22$, $0.21$ and $0.20$ years at $N = 10^3$, and $0.015$, $0.007$ and $0.006$ years at $N = 10^6$. The Gompertz model is the true model in this scenario. In the largest populations, the remaining error of the RW2 posterior mean is mostly the bias of the interpolation \eqref{eq:interp} ($-0.013$ years at $N = 10^6$, against $-0.0125$ years for the noise-free schedule), which the P-spline and Gompertz estimators avoid by maximizing a continuous density. With a secondary peak the Gompertz model is misspecified, and its estimates are biased by $-3.1$, $-5.2$ and $-7.7$ years in scenarios (b) to (d), whatever the population size. After the half-year shift, DDT is nearly unbiased except at $N = 10^3$ ($-0.26$ years in the unimodal scenario and $-1.18$ in scenario (d)). Its estimates vary more across data sets. In the unimodal scenario its mean absolute error is $1.54$, $0.80$, $0.33$ and $0.12$ years from $N = 10^3$ to $10^6$, seven to twelve times that of the RW2 posterior mean. With a secondary peak, the P-spline estimator is the most accurate: at $N = 10^4$ its mean absolute error is $0.20$, $0.21$ and $0.24$ years in scenarios (b) to (d), against $0.31$, $0.35$ and $0.39$ for the RW2 posterior mean. The difference is smaller at $N = 10^3$ ($0.42$ against $0.50$ in scenario (b)).

\begin{figure}[htb!]
    \centering
    \includegraphics[width=\linewidth]{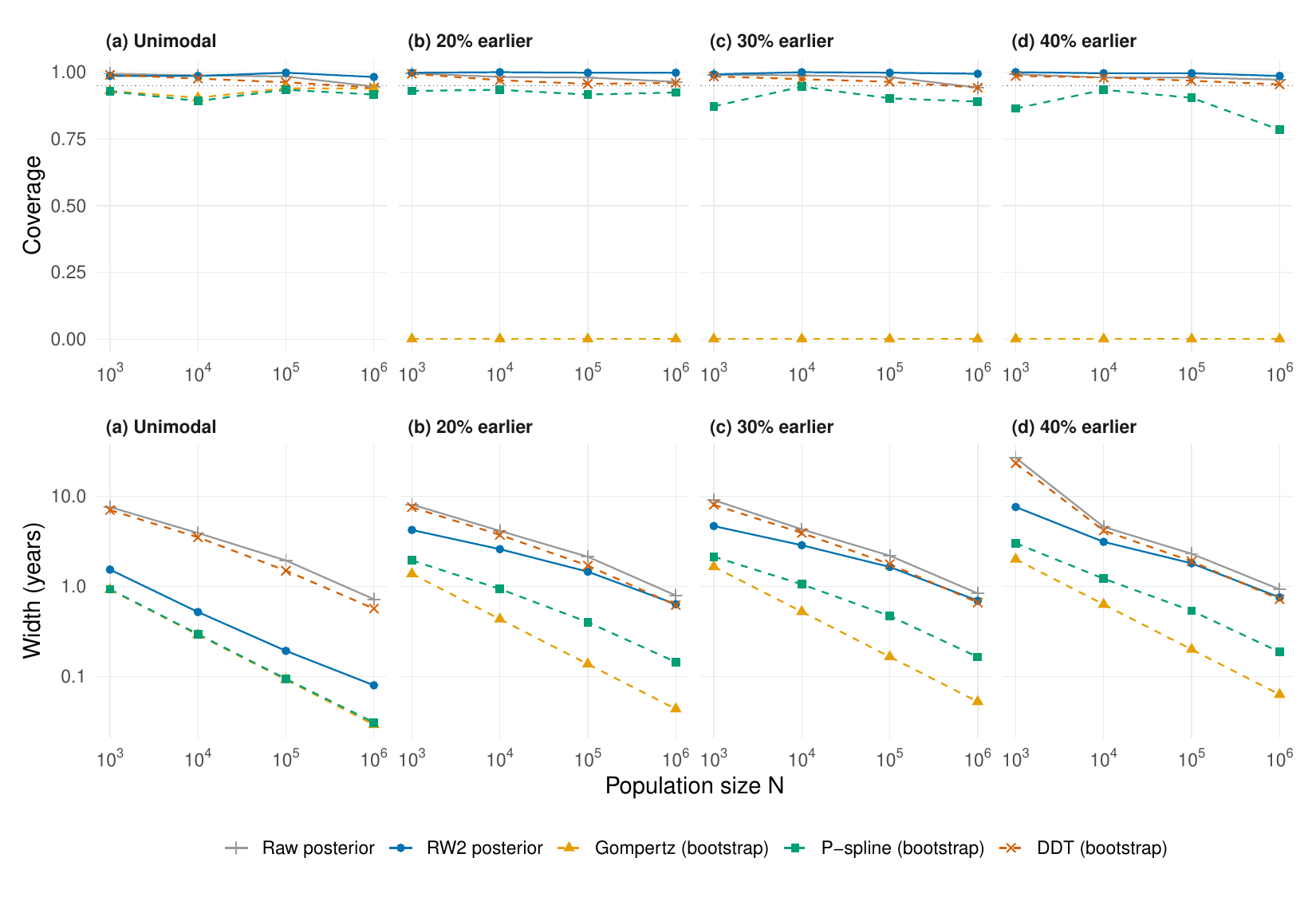}
    \caption{Coverage (top) and mean width (bottom, logarithmic scale) of 95\% intervals for the modal age, by population size and scenario: credible intervals of the raw and RW2 posteriors (solid lines) and bootstrap confidence intervals of the Gompertz, P-spline and DDT estimators (dashed lines). The dotted line marks 95\%.}
    \label{fig:mc_coverage}
\end{figure}

\paragraph{Intervals.} The coverage of the raw posterior credible intervals is close to the nominal level ($0.94$ to $1.00$), and that of the RW2 intervals is conservative ($0.98$ to $1.00$; Figure~\ref{fig:mc_coverage}). The RW2 intervals are much narrower than the raw ones, $0.52$ against $3.91$ years in the unimodal scenario at $N = 10^4$. In the unimodal scenario, the Gompertz and P-spline intervals are narrower than the RW2 intervals ($0.29$ against $0.52$ years at $N = 10^4$) and contain the true mode in $89\%$ to $94\%$ of replications, slightly below the nominal level. With a secondary peak, the Gompertz intervals never contain it: they are narrow ($0.04$ to $2.0$ years) and centered on a biased estimate. The P-spline intervals then cover it in $78\%$ to $95\%$ of replications. Holding the smoothing parameter fixed in the bootstrap leaves out its uncertainty \citep{wood2016smoothing}. In a subset of settings, re-selecting it in every bootstrap sample raised the coverage to $0.96$ and $0.97$ in the unimodal scenario and to $0.94$ in scenario (d) at $N = 10^3$, with intervals 1.6 to 2 times as wide. Coverage remained low in scenario (d) at $N = 10^6$ ($0.80$ against $0.78$; Appendix~\ref{app:implementation}). The DDT intervals have coverage near or above the nominal level ($0.94$ to $0.99$) and are wide, $4.6$ to $7.8$ times wider than the RW2 intervals in the unimodal scenario ($3.50$ against $0.52$ years at $N = 10^4$). Only the RW2 credible intervals stayed at or above the nominal level in every scenario and population size. With a secondary peak, the RW2 coverage is conservative: its intervals ($2.6$ to $3.1$ years wide at $N = 10^4$) are wider than the errors require, with coverage close to one and mean absolute errors of $0.31$ to $0.39$ years. One smoothing parameter must fit both the bend between the two peaks and the main peak, which suggests that the marginal likelihood favors weaker smoothing than the main peak alone would need.

\section{Application}
\label{sec:application}

\paragraph{Data and design.} We use deaths and exposures by single year of age ($0, \dots, 109$ and $110+$) from the \citet{hmd} for women in Denmark, France, Japan and Sweden, from 1960 to the latest year available (2025 for Denmark and 2024 for the other countries), 261 population-years in total. For each population-year we estimate the modal age with the RW2 posterior and with the Gompertz, P-spline and DDT estimators, each with the starting ages $x_0 = 40$, $50$, $60$ and $70$. The modal age lies above all four, so the target is the same for every starting age (Section~\ref{sec:target}), and differences across starting ages measure how sensitive each estimator is to the age range it is given. DDT uses the rates at each age separately and does not depend on $x_0$; the other three fit a model across ages.

\begin{figure*}[htb!]
    \centering
    \includegraphics[width=\linewidth]{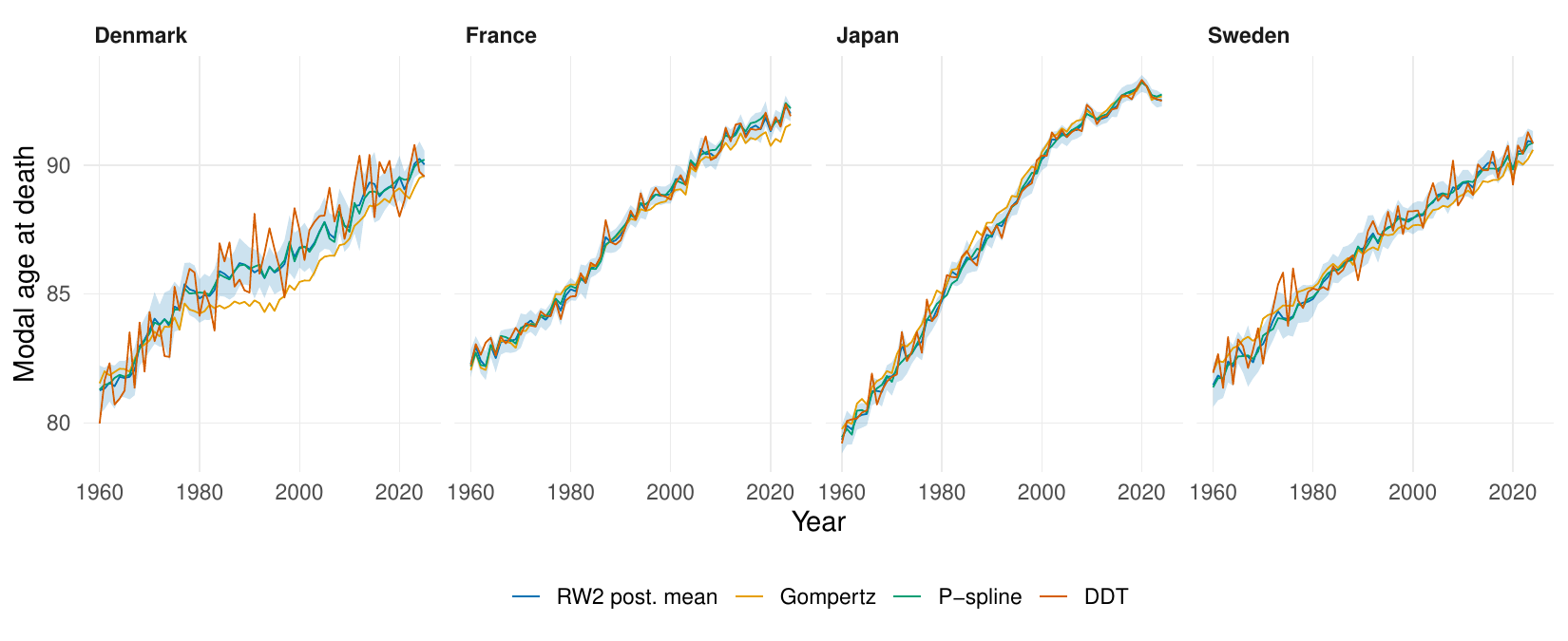}
    \caption{Modal age at death of women by country and year, starting age 40. The band is the 95\% credible interval of the RW2 posterior.}
    \label{fig:emp_trends}
\end{figure*}

\paragraph{Trends and agreement between methods.} Between 1960 and the latest year, the RW2 posterior mean rose by $8.8$ years in Denmark ($81.3$ to $90.0$), $9.4$ in Sweden ($81.5$ to $90.9$), $9.7$ in France ($82.3$ to $92.0$) and $13.2$ in Japan ($79.4$ to $92.5$) (Figure~\ref{fig:emp_trends}). The median width of the 95\% credible interval is $1.1$ years, ranging from $0.5$ to $3.3$ years; it is widest in Denmark (median $1.6$ years), the smallest population, and narrowest in France and Japan ($0.8$ years). The P-spline estimate lies within the RW2 credible interval in $99.6\%$ of population-years, with a mean absolute difference of $0.12$ years from the RW2 posterior mean. The Gompertz estimate lies within it in $71.6\%$ of population-years, on average $0.15$ years lower. The DDT estimate lies within it in $81.2\%$ of population-years, with a mean absolute difference of $0.39$ years, and its year-to-year series is rougher because DDT does not smooth the single-year rates (Figure~\ref{fig:emp_start_age}). These percentages measure agreement between estimators applied to the same data, not accuracy.

\paragraph{Change between 2019 and 2020.} To assess the change in modal age, we pair independent posterior draws of $\tilde{M}$ for 2019 and 2020, which are fitted separately, and compute their difference. Marginal credible intervals can overlap even when the difference is credibly different from zero \citep{schenker2001judging}. Table~\ref{tab:change} and Figure~\ref{fig:emp_change} show the result. In France the modal age declined with posterior probability $0.99$: the change was $-0.52$ years, with a 95\% credible interval of $-1.24$ to $-0.09$, although the marginal intervals of the two years overlap. In Sweden the posterior mean fell by $0.54$ years, and the probability of a decline was $0.94$, with an interval for the change that includes zero ($-1.25$ to $0.16$). In Denmark and Japan the posterior mean rose, and the probability of a decline was $0.22$ and $0.19$. The most probable single-year age was the same in both years in Denmark, France and Japan, and moved from 90 to 89 in Sweden.

\begin{table}[htb!]
\centering
\caption{Modal age at death of women in 2019 and 2020 (RW2 posterior, starting age 40): most probable single-year age with its posterior probability, posterior mean of the change in $\tilde{M}$ from 2019 to 2020 with 95\% credible interval, and posterior probability that the modal age declined.}
\label{tab:change}
\small
\begin{tabular}{lcccc}
\hline
 & \multicolumn{2}{c}{Most probable age (probability)} & Change, 2019 to 2020 & \\
Country & 2019 & 2020 & mean (95\% interval) & P(decline) \\
\hline
Denmark & 89 (0.6018) & 89 (0.7065) & 0.4471 ($-$0.7144, 1.5916) & 0.2154 \\
France & 91 (0.8269) & 91 (0.9508) & $-$0.5168 ($-$1.2422, $-$0.0932) & 0.9932 \\
Japan & 93 (0.5567) & 93 (0.8669) & 0.2029 ($-$0.2509, 0.6294) & 0.1851 \\
Sweden & 90 (0.8882) & 89 (0.7317) & $-$0.5360 ($-$1.2509, 0.1599) & 0.9353 \\
\hline
\end{tabular}
\end{table}

\begin{figure*}[htb!]
    \centering
    \includegraphics[width=\linewidth]{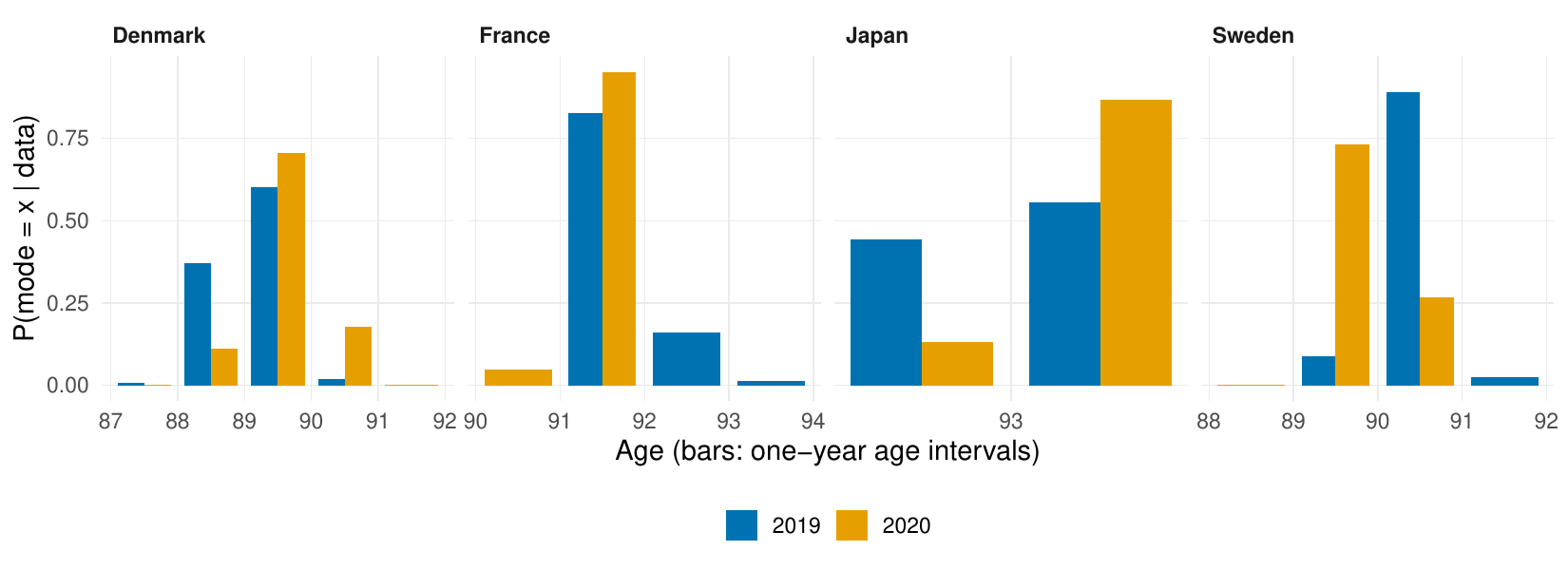}
    \caption{Posterior distribution of the modal age at death of women in 2019 and 2020 (RW2 prior, starting age 40): probability that each single-year age interval contains the mode.}
    \label{fig:emp_change}
\end{figure*}

\begin{figure*}[htb!]
    \centering
    \includegraphics[width=\linewidth]{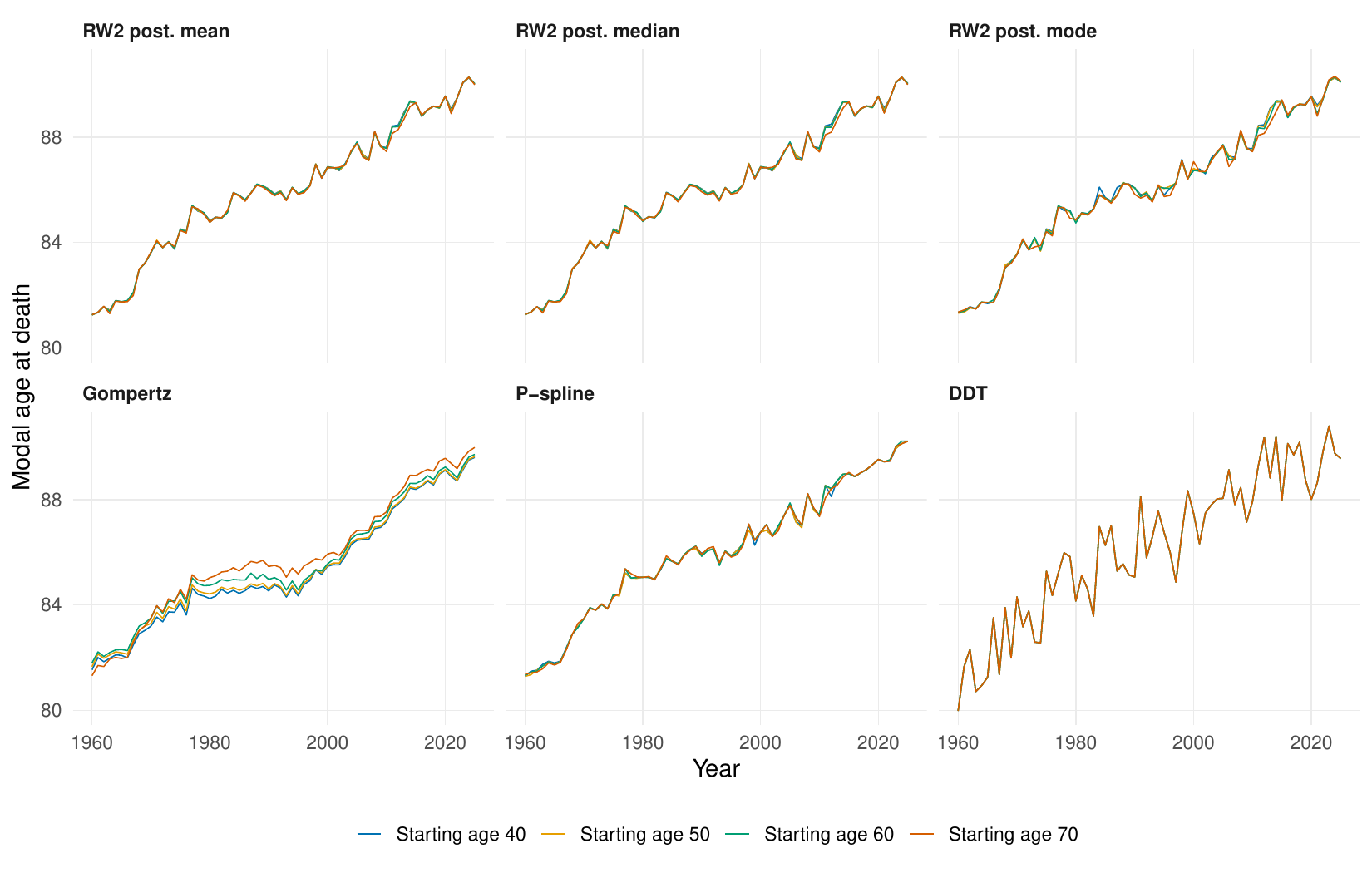}
    \caption{Modal age at death of women in Denmark by year, for each method (panels) and starting age (colors). Where the four lines coincide, the estimate does not depend on the starting age.}
    \label{fig:emp_start_age}
\end{figure*}

\paragraph{Sensitivity to the starting age.} Figure~\ref{fig:emp_start_age} shows the estimates for Denmark, the smallest of the four populations, and Table~\ref{tab:start_age} summarizes all four (Figure~\ref{fig:app_start_all} in the Appendix shows each country). The RW2 posterior mean and the P-spline estimator are stable across starting ages. Moving it from 40 to 70 changes the RW2 posterior mean by $0.03$ years on average and by at most $0.28$ years, against $0.04$ and $0.48$ years for the P-spline estimator and $0.42$ and $0.99$ years for the Gompertz model. The Gompertz estimates also shift systematically, by $+0.37$ years on average, because the fitted slope depends on the age range; in Figure~\ref{fig:emp_start_age} this appears as four separate lines. Of the three posterior summaries, the posterior mode is the most sensitive (at most $0.57$ years), in line with its lower accuracy in the simulation. For the RW2 posterior mean, the average absolute change is more than an order of magnitude smaller than the median width of the credible interval. This comparison measures stability, not accuracy: an estimator can be stable and biased, as the Gompertz model is with a secondary peak.

\begin{table}[htb!]
\centering
\caption{Change in the estimated modal age at death when the starting age moves from 40 to 50, 60 and 70: mean absolute change and largest absolute change (years), women in 4 countries, all years (261 population-years). DDT uses the rates at each age separately and does not depend on the starting age.}
\label{tab:start_age}
\small
\begin{tabular}{lrrrrrr}
\hline
 & \multicolumn{2}{c}{50 vs 40} & \multicolumn{2}{c}{60 vs 40} & \multicolumn{2}{c}{70 vs 40} \\
Method & Mean abs. & Max abs. & Mean abs. & Max abs. & Mean abs. & Max abs. \\
\hline
RW2 post. mean & 0.0104 & 0.0931 & 0.0168 & 0.2273 & 0.0280 & 0.2825 \\
RW2 post. median & 0.0113 & 0.1026 & 0.0192 & 0.2496 & 0.0311 & 0.3396 \\
RW2 post. mode & 0.0263 & 0.3041 & 0.0425 & 0.4635 & 0.0521 & 0.5676 \\
P-spline & 0.0128 & 0.2623 & 0.0213 & 0.3043 & 0.0351 & 0.4754 \\
Gompertz & 0.1126 & 0.2421 & 0.2958 & 0.5860 & 0.4154 & 0.9948 \\
DDT & 0.0000 & 0.0000 & 0.0000 & 0.0000 & 0.0000 & 0.0000 \\
\hline
\end{tabular}
\end{table}

\section{Discussion}
\label{sec:conclusion}

The posterior distribution of the modal age expresses how strongly the data support each candidate age. It provides a basis for comparing populations and years when neighboring ages are difficult to distinguish.

For routine estimation, we recommend the RW2 prior and the posterior mean. In simulations, sharing information across ages reduced the mean absolute error by a factor of $1.4$ to $10.6$ relative to independent priors and reduced the downward bias of the raw posterior in small populations. Unlike the Gompertz model, RW2 does not restrict the log rates to a straight line. With a secondary peak, the Gompertz estimates were biased by 3 to 8 years, while the bias of the RW2 posterior mean stayed within $0.10$ years of zero except in the smallest populations. The independent gamma posterior remains useful when no smoothing is desired and permits exact sampling.

RW2 credible intervals were conservative, with coverage of $0.98$ to $1.00$ in every scenario and population size. The P-spline bootstrap with fixed smoothing covered the true mode in as few as $78\%$ of replications; Gompertz intervals never covered it with a secondary peak, and DDT intervals were four to eight times wider than RW2 intervals when the distribution was unimodal. The interval comparison concerns the procedures implemented here. RW2 integrates over smoothing uncertainty, whereas the main P-spline bootstrap fixes the smoothing parameter. Re-selecting that parameter improved P-spline coverage in most of the settings checked and widened the intervals, limiting any general ranking of the two methods.

The French decline in 2020 illustrates why the posterior of a difference is useful: the annual marginal intervals overlap, yet the posterior supports a decline. In small populations, the full distribution also identifies cases in which several ages remain plausible modes. In large populations, where the RW2 and P-spline point estimates agree closely, the main additional information is the uncertainty attached to the estimate. The RW2 posterior mean was also stable across starting ages: moving the starting age from 40 to 70 changed it by $0.03$ years on average, against $0.42$ years for Gompertz.

A limitation of the RW2 model is its single smoothing parameter, which must fit both the bend between peaks and the main peak. With a secondary peak, the RW2 posterior mean was less accurate than the P-spline estimator ($0.31$ to $0.39$ against $0.20$ to $0.24$ years at $N = 10^4$), and the RW2 credible intervals were wider than needed. A smoothing precision that varies with age in \eqref{eq:rw2} could improve local adaptation, but would need validation. The simulation covered one family of Gompertz mixtures, and the intervals are conditional on the exposures and the Poisson model; they do not include errors in exposures or reported ages. The same posterior draws can be used for other life-table quantities, such as life expectancy or lifespan variation.

\section*{Acknowledgments}

\paragraph*{Funding:}
We gratefully acknowledge the financial support from the AXA Research Fund through the funding for the “\textit{AXA Chair in Longevity Research}”. This article is funded by the European Union (ERC, Born Once – Die Once, Grant agreement ID 101043983). Views and opinions expressed are however those of the author only and do not necessarily reflect those of the European Union or the European Research Council Executive Agency. Neither the European Union nor the granting authority can be held responsible for them.


\paragraph*{Competing interests:} The authors declare that they have no known competing financial interests or personal relationships that could have appeared to influence the work reported in this paper.

\paragraph*{Data and materials availability:} The code, simulation settings, and materials required to reproduce the analyses, tables, and figures are openly available at \url{https://github.com/scpatricio/modal_age_posterior}.

{\footnotesize
\bibliographystyle{apalike}
\bibliography{references}

@book{bishop1975discrete,
  title={Discrete multivariate analysis theory and practice},
  author={Bishop, Yvonne MM and Fienberg, Stephen E and Holland, Paul W},
  year={1975},
  publisher={Springer}
}

@article{bergeron2015decomposing,
  title={Decomposing changes in life expectancy: Compression versus shifting mortality},
  author={Bergeron-Boucher, Marie-Pier and Ebeling, Marcus and Canudas-Romo, Vladimir},
  journal={Demographic Research},
  volume={33},
  pages={391--424},
  year={2015},
  publisher={JSTOR}
}

@book{keyfitz2005applied,
  title={Applied mathematical demography},
  author={Keyfitz, Nathan and Caswell, Hal},
  edition={3rd},
  year={2005},
  publisher={Springer}
}

@article{andreev2010spreadsheet,
  title={Spreadsheet for calculation of confidence limits for any life table or healthy-life table quantity},
  author={Andreev, Evgueni M and Shkolnikov, Vladimir M},
  journal={Rostock, Germany: Max Planck Institute for Demographic Research},
  pages={2013--2017},
  year={2010}
}

@article{missov2015gompertz,
  title={The Gompertz force of mortality in terms of the modal age at death},
  author={Missov, Trifon I and Lenart, Adam and Nemeth, Laszlo and Canudas-Romo, Vladimir and Vaupel, James W},
  journal={Demographic Research},
  volume={32},
  pages={1031--1048},
  year={2015},
  publisher={JSTOR}
}

@article{vazquez2024longevity,
  title={Longevity {\`a} la mode},
  author={Vazquez-Castillo, Paola and Bergeron-Boucher, Marie-Pier and Missov, Trifon I},
  journal={Demographic Research},
  volume={50},
  pages={325--346},
  year={2024},
  publisher={JSTOR}
}

@article{horiuchi2013modal,
  title={Modal age at death: lifespan indicator in the era of longevity extension},
  author={Horiuchi, Shiro and Ouellette, Nadine and Cheung, Siu Lan Karen and Robine, Jean-Marie},
  journal={Vienna Yearbook of Population Research},
  pages={37--69},
  year={2013},
  publisher={JSTOR}
}

@article{kannisto2001mode,
  title={Mode and dispersion of the length of life},
  author={Kannisto, V{\"a}in{\"o}},
  journal={Population: An English Selection},
  pages={159--171},
  year={2001},
  publisher={JSTOR}
}

@article{canudas2008modal,
  title={The modal age at death and the shifting mortality hypothesis},
  author={Canudas-Romo, Vladimir},
  journal={Demographic Research},
  volume={19},
  pages={1179--1204},
  year={2008},
  publisher={JSTOR}
}

@article{ouellette2011changes,
  title={Changes in the age-at-death distribution in four low mortality countries: A nonparametric approach},
  author={Ouellette, Nadine and Bourbeau, Robert},
  journal={Demographic Research},
  volume={25},
  pages={595--628},
  year={2011},
  publisher={JSTOR}
}

@book{genz2009computation,
  title={Computation of multivariate normal and t probabilities},
  author={Genz, Alan and Bretz, Frank},
  year={2009},
  publisher={Springer Science \& Business Media}
}

@article{brillinger1986biometrics,
  title={A biometrics invited paper with discussion: the natural variability of vital rates and associated statistics},
  author={Brillinger, David R},
  journal={Biometrics},
  pages={693--734},
  year={1986},
  publisher={JSTOR}
}

@misc{hmd,
author = {{Human Mortality Database}},
year = {2026},
title = {Human Mortality Database},
note = {Max Planck Institute for Demographic Research (Germany), University of California, Berkeley (USA), and French Institute for Demographic Studies (France). Available at \url{https://www.mortality.org}. Data downloaded on October 28, 2026}
}

@article{mi2009mvtnorm,
  title={mvtnorm: New numerical algorithm for multivariate normal probabilities},
  author={Mi, Xuefei and Miwa, Tetsuhisa and Hothorn, Torsten},
  journal={R Journal 1 (2009), Nr. 1},
  volume={1},
  pages={37--39},
  year={2009},
  publisher={Vienna: The R Foundation}
}

@book{preston2000demography,
  title     = {Demography: Measuring and Modeling Population Processes},
  author    = {Preston, Samuel H. and Heuveline, Patrick and Guillot, Michel},
  year      = {2000},
  publisher = {Wiley-Blackwell},
  address   = {Malden, MA},
  isbn      = {978-1557864512}
}

@book{gelman1995bayesian,
  title={Bayesian data analysis},
  author={Gelman, Andrew and Carlin, John B and Stern, Hal S and Rubin, Donald B},
  year={1995},
  publisher={chapman and hall/cRc}
}

@article{whittaker1922new,
  title={On a new method of graduation},
  author={Whittaker, Edmund T},
  journal={Proceedings of the Edinburgh Mathematical Society},
  volume={41},
  pages={63--75},
  year={1922},
  publisher={Cambridge University Press}
}

@article{eilers1996flexible,
  title={Flexible smoothing with B-splines and penalties},
  author={Eilers, Paul HC and Marx, Brian D},
  journal={Statistical science},
  volume={11},
  number={2},
  pages={89--121},
  year={1996},
  publisher={Institute of Mathematical Statistics}
}

@article{camarda2012mortalitysmooth,
  title={MortalitySmooth: An R package for smoothing Poisson counts with P-splines},
  author={Camarda, Carlo G},
  journal={Journal of Statistical Software},
  volume={50},
  pages={1--24},
  year={2012}
}

@article{lang2004bayesian,
  title={Bayesian P-splines},
  author={Lang, Stefan and Brezger, Andreas},
  journal={Journal of computational and graphical statistics},
  volume={13},
  number={1},
  pages={183--212},
  year={2004},
  publisher={Taylor \& Francis}
}

@book{rue2005gaussian,
  title={Gaussian Markov random fields: theory and applications},
  author={Rue, Havard and Held, Leonhard},
  year={2005},
  publisher={Chapman and Hall/CRC}
}

@article{rue2009approximate,
  title={Approximate Bayesian inference for latent Gaussian models by using integrated nested Laplace approximations},
  author={Rue, H{\aa}vard and Martino, Sara and Chopin, Nicolas},
  journal={Journal of the royal statistical society: Series b (statistical methodology)},
  volume={71},
  number={2},
  pages={319--392},
  year={2009},
  publisher={Wiley Online Library}
}

@article{tierney1986accurate,
  title={Accurate approximations for posterior moments and marginal densities},
  author={Tierney, Luke and Kadane, Joseph B},
  journal={Journal of the american statistical association},
  volume={81},
  number={393},
  pages={82--86},
  year={1986},
  publisher={Taylor \& Francis}
}

@article{gompertz1825nature,
  title={On the nature of the function expressive of the law of human mortality, and on a new mode of determining the value of life contingencies},
  author={Gompertz, Benjamin},
  journal={Philosophical Transactions of the Royal Society of London},
  volume={115},
  pages={513--583},
  year={1825}
}

@article{andreev2015confidence,
  title={Confidence estimation of demographic rates on example of mortality rates},
  author={Andreev, EM and Jdanov, DA and Jasilionis, DOMANTAS},
  journal={Demographic Review},
  volume={2},
  number={1},
  pages={24--55},
  year={2015}
}

@article{wood2016smoothing,
  title={Smoothing parameter and model selection for general smooth models},
  author={Wood, Simon N and Pya, Natalya and S{\"a}fken, Benjamin},
  journal={Journal of the American Statistical Association},
  volume={111},
  number={516},
  pages={1548--1563},
  year={2016},
  publisher={Taylor \& Francis}
}

@article{morris2019using,
  title={Using simulation studies to evaluate statistical methods},
  author={Morris, Tim P and White, Ian R and Crowther, Michael J},
  journal={Statistics in medicine},
  volume={38},
  number={11},
  pages={2074--2102},
  year={2019},
  publisher={Wiley Online Library}
}

@article{schenker2001judging,
  title={On judging the significance of differences by examining the overlap between confidence intervals},
  author={Schenker, Nathaniel and Gentleman, Jane F},
  journal={The American Statistician},
  volume={55},
  number={3},
  pages={182--186},
  year={2001},
  publisher={Taylor \& Francis}
}
}

\clearpage



\renewcommand{\thefigure}{S\arabic{figure}}
\renewcommand{\thetable}{S\arabic{table}}
\renewcommand{\theequation}{S\arabic{equation}}
\renewcommand{\thepage}{S\arabic{page}}
\setcounter{figure}{0}
\setcounter{table}{0}
\setcounter{equation}{0}
\setcounter{page}{1} 

\appendix
\section*{Appendix}

\section{Local modes and a Gaussian approximation}
\label{app:analytic}

Age $x$ is a local mode if $d_{x-1} < d_x > d_{x+1}$. Equation~\eqref{eq:lifetable} gives $d_{x-1} = \ell_{x-1}(1 - e^{-m_{x-1}})$, $d_x = \ell_{x-1} e^{-m_{x-1}}(1 - e^{-m_x})$ and $d_{x+1} = \ell_{x-1} e^{-m_{x-1}} e^{-m_x} (1 - e^{-m_{x+1}})$. The common factors cancel, and each comparison involves only two neighboring rates:
\begin{align}
d_x > d_{x-1} &\iff m_{x-1} < u(m_x), \label{eq:left}\\
d_x > d_{x+1} &\iff m_{x+1} < v(m_x), \label{eq:right}
\end{align}
where $u(m) = \log(2 - e^{-m})$, and $v(m) = -\log(2 - e^{m})$ for $m < \log 2$ and $v(m) = +\infty$ otherwise. Under \eqref{eq:gamma_post}, $m_{x-1}$, $m_x$ and $m_{x+1}$ are independent. Conditioning on $m_x$ therefore makes the two events independent, and integrating over $m_x$ gives the following result.

\begin{proposition}
\label{prop:local}
Let $g_x$ and $G_x$ denote the posterior density and distribution function of $m_x$ in \eqref{eq:gamma_post}, with $G_x(+\infty) = 1$. Then
\begin{equation}
\mathbb{P}(d_{x-1} < d_x > d_{x+1} \mid \bm{D}) 
 = \int_0^\infty g_x(m)\, G_{x-1}\big(u(m)\big)\, G_{x+1}\big(v(m)\big)\, dm .
\label{eq:local}
\end{equation}
\end{proposition}

The integral in \eqref{eq:local} is one-dimensional and can be evaluated to high accuracy with standard quadrature. For small rates, $u(m) \approx m - m^2$ and $v(m) \approx m + m^2$, which gives the condition for a local mode stated in Section~\ref{sec:analytic}.

\paragraph{Gaussian approximation.} When many deaths are observed at each age, the posterior of $\bm{m}$ is approximately Gaussian with mean $\hat{\bm{m}}$, the vector of posterior means $(\alpha + D_x)/(\beta + E_x)$, and diagonal covariance $\bm{V}$ with entries $(\alpha + D_x)/(\beta + E_x)^2$. By the delta method, $\bm{d} \approx \mathcal{N}(\hat{\bm{d}}, \bm{\Sigma}_d)$ with $\hat{\bm{d}} = \bm{d}(\hat{\bm{m}})$ and $\bm{\Sigma}_d = \bm{J} \bm{V} \bm{J}^\top$. The Jacobian $\bm{J}$ is lower triangular:
\begin{equation}
\frac{\partial d_x}{\partial m_y} =
\begin{cases}
-d_x, & y < x,\\
\ell_{x+1}, & y = x,\\
0, & y > x,
\end{cases}
\label{eq:jacobian}
\end{equation}
evaluated at $\hat{\bm{m}}$. A higher rate at an earlier age lowers $d_x$ by lowering survival to $x$; a higher rate at age $x$ raises $d_x$; rates at later ages do not affect it.

Age $x$ is the mode when all differences $d_x - d_j$, $j \neq x$, are positive. Let $\bm{\Lambda}_x$ be the $(K-1) \times K$ matrix whose rows are $\bm{e}_x - \bm{e}_j$ for $j \neq x$, where $\bm{e}_k$ is the $k$-th unit vector. Then $\bm{Y}_x = \bm{\Lambda}_x \bm{d} \approx \mathcal{N}(\bm{\Lambda}_x \hat{\bm{d}}, \, \bm{\Lambda}_x \bm{\Sigma}_d \bm{\Lambda}_x^\top)$, and
\begin{equation}
\mathbb{P}(M = x \mid \bm{D}) \approx \int_{R_x} f_{\bm{Y}_x}(\bm{y}) \, d\bm{y},
\label{eq:gauss}
\end{equation}
where $f_{\bm{Y}_x}$ is the Gaussian density of $\bm{Y}_x$ and $R_x = [0, \infty)^{K-1}$ is the positive orthant. Orthant probabilities of this kind can be computed with the algorithms of \citet{genz2009computation}, implemented in the \texttt{mvtnorm} package \citep{mi2009mvtnorm}. The Genz--Bretz algorithm uses randomized quasi-Monte Carlo integration. Equation~\eqref{eq:gauss} helps interpret uncertainty in the mode through the standardized gaps $(\hat d_x - \hat d_j)/\text{sd}(d_x - d_j)$: two ages are both plausible modes when their life-table deaths differ by little relative to the posterior uncertainty of that difference. The approximation requires many deaths at the ages near the mode. When counts are small, \eqref{eq:mc} should be used.

\section{Proof of Proposition~\ref{prop:convergence}}
\label{app:proof}

\begin{proof}
Write $\bar m_x = (\alpha + D_x)/(\beta + c_x N)$ and $s_x^2 = (\alpha + D_x)/(\beta + c_x N)^2$ for the posterior mean and variance of $m_x$. Since $D_x / (c_x N) \to m_x^0$ in probability, $\bar m_x \to m_x^0$ and $s_x^2 \to 0$ in probability. Fix $\eta > 0$. On the event $|\bar m_x - m_x^0| < \eta/2$, Chebyshev's inequality gives $\mathbb{P}(|m_x - m_x^0| > \eta \mid \bm{D}) \le 4 s_x^2 / \eta^2$, and the probability of that event tends to one. A union bound over the $K$ ages then gives $\mathbb{P}(\max_x |m_x - m_x^0| > \eta \mid \bm{D}) \to 0$ in probability.

The map $\bm{m} \mapsto \bm{d}(\bm{m})$ in \eqref{eq:lifetable} is continuous, so we can choose $\eta$ such that $\max_x |m_x - m_x^0| \le \eta$ implies $|d_x(\bm{m}) - d_x(\bm{m}^0)| < \delta/2$ for all $x$. For such $\bm{m}$, $d_{x^\star}(\bm{m}) > d_{x^\star}(\bm{m}^0) - \delta/2 \ge d_j(\bm{m}^0) + \delta/2 > d_j(\bm{m})$ for every $j \neq x^\star$, so $M(\bm{m}) = x^\star$. Hence $\mathbb{P}(M = x^\star \mid \bm{D}) \ge 1 - \mathbb{P}(\max_x |m_x - m_x^0| > \eta \mid \bm{D}) \to 1$.
\end{proof}

\section{The bootstrap distribution of the mode}
\label{app:bootstrap}

The parametric bootstrap of Section~\ref{sec:why_posterior} reports the frequencies
\begin{equation}
\mathbb{P}^\ast(\hat M^\ast = x) = \sum_{\bm{D}^\ast} \mathbb{I}\{M(\bm{D}^\ast/\bm{E}) = x\}\, p(\bm{D}^\ast \mid \hat{\bm{m}}).
\label{eq:boot}
\end{equation}
where $p(\bm{D}^\ast \mid \hat{\bm{m}})$ is the Poisson distribution of replicate counts given the estimated rates. This sum evaluates the sampling probabilities described in Section~\ref{sec:why_posterior}.

With sparse data, the plug-in scheme used here can differ substantially from the posterior because it resamples from the observed rates. If no deaths are observed at some age, then $\hat m_x = 0$ and every replicate also has zero deaths at that age, so this bootstrap treats the rate as known to be zero. The posterior \eqref{eq:gamma_post} keeps a proper distribution, $\text{Gamma}(1/2, E_x)$, which assigns positive probability to positive rates. A bootstrap from a fitted smooth schedule would not have this problem.

\section{Computation of the RW2 posterior}
\label{app:computation}

For fixed $\lambda$, the posterior mode $\hat{\bm{\eta}}_\lambda$ of Section~\ref{sec:rw2} is found by penalized Fisher scoring, iterating
\begin{equation}
(\bm{W} + \lambda \bm{P})\, \bm{\eta}^{\text{new}} = \bm{W} \bm{\eta} + \bm{D} - \bm{\mu},
\qquad \mu_x = E_x e^{\eta_x}, \quad \bm{W} = \operatorname{diag}(\bm{\mu}),
\label{eq:scoring}
\end{equation}
until the largest change in $\bm{\eta}$ is below $10^{-7}$. At the mode, $\bm{H}_\lambda = \bm{W}(\hat{\bm{\eta}}_\lambda) + \lambda \bm{P}$. The matrix $\bm{H}_\lambda$ is banded with bandwidth two, so each iteration, the Cholesky factor of $\bm{H}_\lambda$ and each draw from $\mathcal{N}(\hat{\bm{\eta}}_\lambda, \bm{H}_\lambda^{-1})$ cost $O(K)$ operations. Each grid value of $\lambda$ is weighted by the Laplace approximation of its marginal likelihood,
\begin{equation}
\log p(\bm{D} \mid \lambda) \approx \ell(\hat{\bm{\eta}}_\lambda) - \frac{\lambda}{2}\, \hat{\bm{\eta}}_\lambda^\top \bm{P} \hat{\bm{\eta}}_\lambda + \frac{K-2}{2} \log \lambda - \frac{1}{2} \log \lvert \bm{H}_\lambda \rvert + \text{const}.
\label{eq:marglik}
\end{equation}

\paragraph{Comparison with posterior sampling.} To check the approximation, we sampled the posterior of the same model (Poisson likelihood, RW2 prior, flat prior on the grid of $\log_{10}\lambda$) for women in Denmark, France, Japan and Sweden in 1960 and 2019, with starting age 50. We also sampled four small simulated populations from Section~\ref{sec:simulation} (scenarios (a) and (d), $N = 10^3$, two replications each, starting age 40). We used an independence Metropolis--Hastings sampler whose proposal is the Laplace approximation itself: a grid value $\lambda_j$ is drawn with probability $w_j$, and $\bm{\eta}$ is drawn from $\mathcal{N}(\hat{\bm{\eta}}_{\lambda_j}, \bm{H}_{\lambda_j}^{-1})$. For the sampler, every grid value is evaluated and proposed with positive probability ($w_j$ mixed with a uniform distribution with weight $0.01$), so the chain targets the posterior on the full grid. The accept--reject step corrects the error of the approximation, and the acceptance rate measures how close the approximation is: it would be one if the approximation were exact. Each chain has $10^5$ iterations.

In the eight national populations, the acceptance rates are between $0.93$ and $0.98$ (Table~\ref{tab:app_mcmc}), and the effective sample size of $\tilde{M}$ is at least $83{,}000$ of the $10^5$ draws. The posterior means of $\tilde{M}$ from the two computations differ by at most $0.006$ years and the limits of the 95\% credible intervals by at most $0.03$ years. The most probable age is the same in all eight cases, and its posterior probability differs by at most $0.005$ (Figure~\ref{fig:app_mcmc}); the posterior means of $\log_{10}\lambda$ differ by at most $0.002$. In the four small simulated populations, the acceptance rates are lower, $0.88$ to $0.97$, and the posterior means differ by at most $0.05$ years and the interval limits by at most $0.09$ years, less than a tenth of the interval widths ($1.3$ to $4.5$ years). The most probable age is again the same. In these comparisons, the approximation error is small relative to posterior uncertainty in both the national and small simulated populations.

\begin{table}[htb!]
\centering
\caption{RW2 posterior of the modal age at death computed with the Laplace approximation and by posterior sampling (independence Metropolis--Hastings): posterior mean of $\tilde{M}$ with 95\% credible interval, and acceptance rate of the sampler. Women, starting age 50; simulated populations (scenarios a and d of Section~\ref{sec:simulation}, $N = 10^3$), starting age 40.}
\label{tab:app_mcmc}
\small
\begin{tabular}{llccc}
\hline
 & & \multicolumn{2}{c}{Modal age at death} & Acceptance \\
Population & Year & Laplace & Sampling & rate \\
\hline
Denmark & 1960 & 81.2632 (80.2467, 82.2879) & 81.2572 (80.2150, 82.2818) & 0.9659 \\
Denmark & 2019 & 89.1060 (88.3343, 89.9458) & 89.1050 (88.3350, 89.9491) & 0.9638 \\
France & 1960 & 82.3010 (81.7822, 82.9051) & 82.3016 (81.7813, 82.9020) & 0.9784 \\
France & 2019 & 91.8364 (91.5432, 92.4575) & 91.8390 (91.5439, 92.4658) & 0.9254 \\
Japan & 1960 & 79.3633 (78.8526, 79.8683) & 79.3604 (78.8393, 79.8677) & 0.9782 \\
Japan & 2019 & 93.0098 (92.7125, 93.2728) & 93.0109 (92.7167, 93.2775) & 0.9786 \\
Sweden & 1960 & 81.4744 (80.5841, 82.3046) & 81.4767 (80.5745, 82.3039) & 0.9686 \\
Sweden & 2019 & 90.3893 (89.8135, 91.0171) & 90.3907 (89.8220, 91.0186) & 0.9683 \\
Simulated (a) & rep. 1 & 89.6449 (88.9667, 90.2585) & 89.6517 (88.9688, 90.2483) & 0.9698 \\
Simulated (a) & rep. 2 & 89.7661 (89.1542, 90.6394) & 89.7726 (89.1596, 90.6486) & 0.9631 \\
Simulated (d) & rep. 1 & 89.6827 (87.1238, 91.5788) & 89.6948 (87.2130, 91.6044) & 0.8878 \\
Simulated (d) & rep. 2 & 89.2358 (87.4075, 91.2168) & 89.2895 (87.4884, 91.2296) & 0.8753 \\
\hline
\end{tabular}
\end{table}

\begin{figure*}[htb!]
    \centering
    \includegraphics[width=\linewidth]{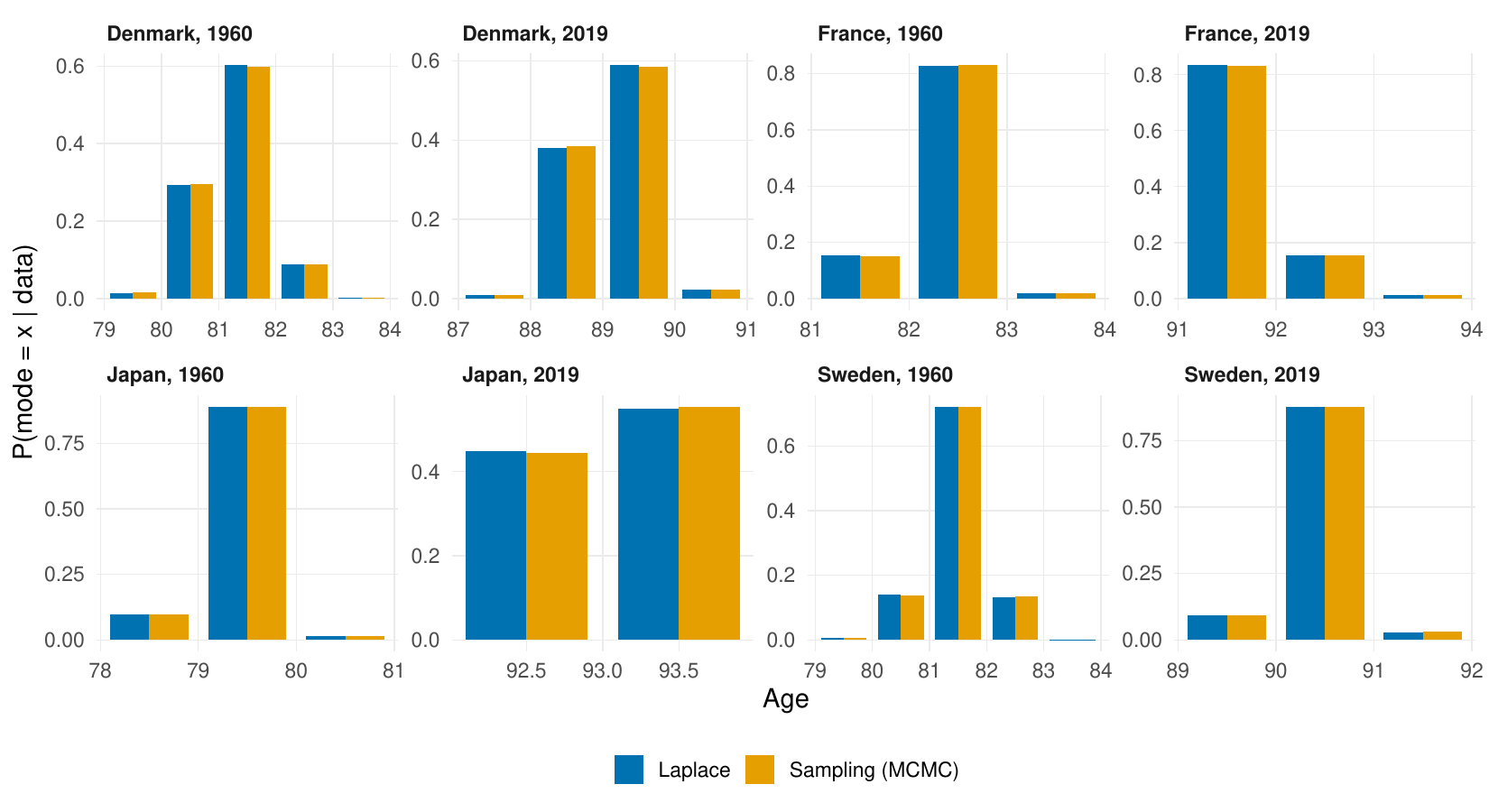}
    \caption{Posterior distribution of the modal age at death under the RW2 prior, computed with the Laplace approximation and by posterior sampling (independence Metropolis--Hastings), women, starting age 50.}
    \label{fig:app_mcmc}
\end{figure*}

\section{Full simulation results}
\label{app:mc}

Tables~\ref{tab:app_mae} to \ref{tab:app_coverage} report the mean absolute error, bias and standard deviation of every estimator, and the coverage and width of the credible and bootstrap intervals, for all scenarios and population sizes of Section~\ref{sec:simulation}. Figure~\ref{fig:mc_bias} shows the bias by population size.

\begin{figure}[htb!]
    \centering
    \includegraphics[width=\linewidth]{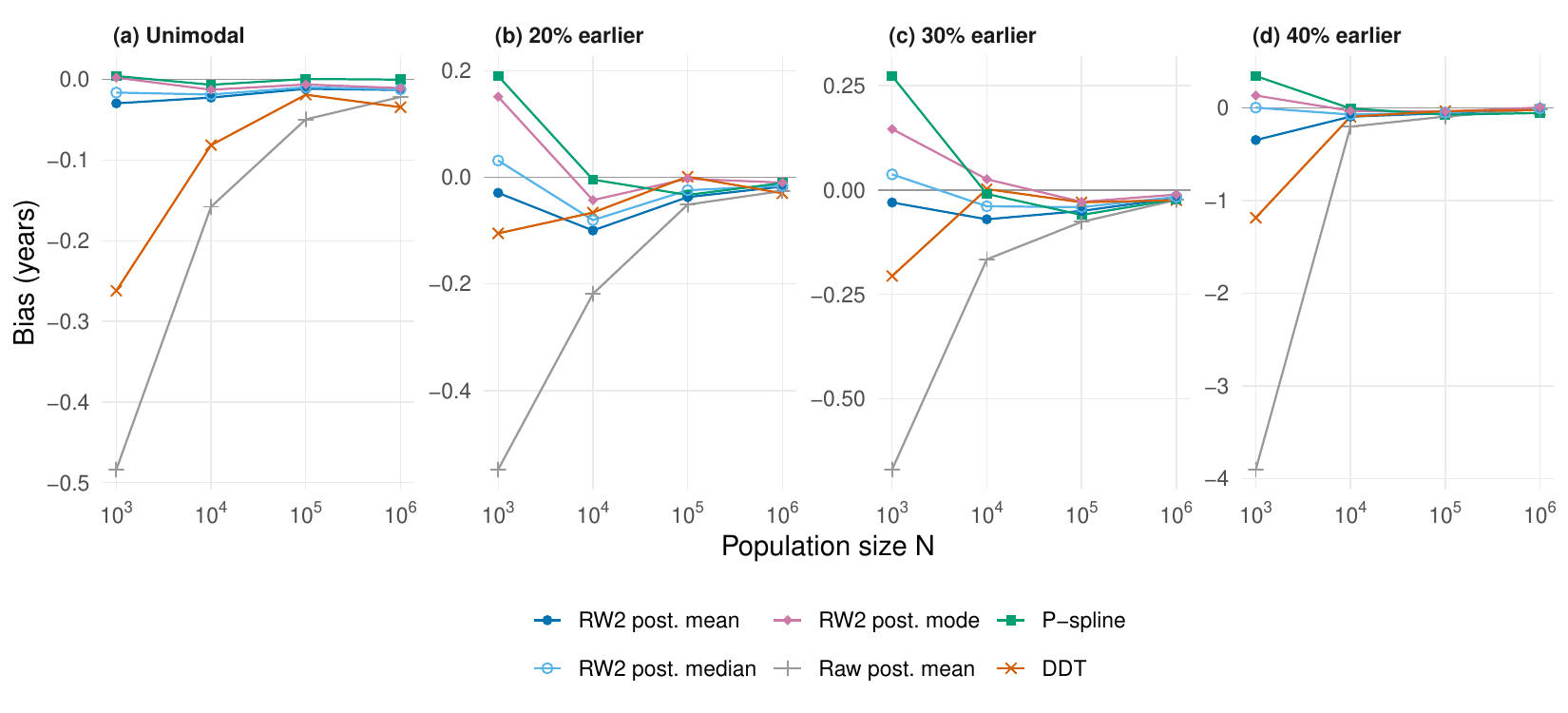}
    \caption{Bias of the estimated modal age by population size and scenario, starting age 40, 500 replications. Gompertz is not shown (its bias is 3 to 8 years with a secondary peak).}
    \label{fig:mc_bias}
\end{figure}

\begin{table}[htb!]
\centering
\caption{Simulation study: mean absolute error (years) of the estimated modal age. Starting age 40, 500 replications, true mode 90.}
\label{tab:app_mae}
\scriptsize
\setlength{\tabcolsep}{4pt}
\begin{tabular}{llrrrrrrrrr}
\hline
 & & \multicolumn{3}{c}{RW2 posterior} & \multicolumn{3}{c}{Raw posterior} & & & \\
Scenario & $N$ & Mean & Median & Mode & Mean & Median & Mode & Gompertz & P-spline & DDT \\
\hline
(a) unimodal & $10^3$ & 0.2227 & 0.2112 & 0.2065 & 0.8435 & 1.1589 & 1.6526 & 0.1967 & 0.2057 & 1.5423 \\
 & $10^4$ & 0.0791 & 0.0750 & 0.0713 & 0.4776 & 0.7004 & 0.9885 & 0.0648 & 0.0697 & 0.8015 \\
 & $10^5$ & 0.0267 & 0.0248 & 0.0225 & 0.2821 & 0.3828 & 0.5145 & 0.0187 & 0.0203 & 0.3321 \\
 & $10^6$ & 0.0148 & 0.0140 & 0.0130 & 0.1508 & 0.1794 & 0.2484 & 0.0060 & 0.0069 & 0.1151 \\[2pt]
(b) 20\% earlier & $10^3$ & 0.5021 & 0.5226 & 0.6232 & 0.8893 & 1.1733 & 1.6887 & 3.1205 & 0.4162 & 1.6085 \\
 & $10^4$ & 0.3145 & 0.3322 & 0.4430 & 0.5239 & 0.7356 & 1.0204 & 3.0901 & 0.1977 & 0.8671 \\
 & $10^5$ & 0.1704 & 0.1892 & 0.2912 & 0.3001 & 0.4140 & 0.5682 & 3.0968 & 0.0837 & 0.3547 \\
 & $10^6$ & 0.0987 & 0.1157 & 0.1755 & 0.1497 & 0.1795 & 0.2563 & 3.0956 & 0.0307 & 0.1179 \\[2pt]
(c) 30\% earlier & $10^3$ & 0.5949 & 0.6282 & 0.7385 & 1.0265 & 1.2772 & 1.7824 & 5.1728 & 0.5337 & 1.6878 \\
 & $10^4$ & 0.3467 & 0.3768 & 0.5131 & 0.5439 & 0.7710 & 1.0843 & 5.1736 & 0.2062 & 0.9078 \\
 & $10^5$ & 0.1980 & 0.2224 & 0.3325 & 0.3155 & 0.4430 & 0.5992 & 5.1765 & 0.1131 & 0.3887 \\
 & $10^6$ & 0.1259 & 0.1465 & 0.2109 & 0.1745 & 0.2060 & 0.2802 & 5.1748 & 0.0387 & 0.1401 \\[2pt]
(d) 40\% earlier & $10^3$ & 0.8204 & 0.7172 & 0.8698 & 3.9555 & 1.9081 & 2.3054 & 7.7035 & 0.6152 & 2.6267 \\
 & $10^4$ & 0.3932 & 0.4400 & 0.6080 & 0.5866 & 0.8455 & 1.1326 & 7.7196 & 0.2439 & 0.9594 \\
 & $10^5$ & 0.2248 & 0.2577 & 0.3860 & 0.3502 & 0.4934 & 0.6506 & 7.7227 & 0.1236 & 0.4465 \\
 & $10^6$ & 0.1220 & 0.1446 & 0.2174 & 0.1686 & 0.2049 & 0.2851 & 7.7208 & 0.0604 & 0.1372 \\
\hline
\end{tabular}
\end{table}

\begin{table}[htb!]
\centering
\caption{Simulation study: bias (years) of the estimated modal age. Starting age 40, 500 replications, true mode 90.}
\label{tab:app_bias}
\scriptsize
\setlength{\tabcolsep}{4pt}
\begin{tabular}{llrrrrrrrrr}
\hline
 & & \multicolumn{3}{c}{RW2 posterior} & \multicolumn{3}{c}{Raw posterior} & & & \\
Scenario & $N$ & Mean & Median & Mode & Mean & Median & Mode & Gompertz & P-spline & DDT \\
\hline
(a) unimodal & $10^3$ & $-$0.0295 & $-$0.0162 & 0.0023 & $-$0.4835 & $-$0.4253 & $-$0.4437 & 0.0104 & 0.0043 & $-$0.2619 \\
 & $10^4$ & $-$0.0224 & $-$0.0186 & $-$0.0127 & $-$0.1577 & $-$0.1372 & $-$0.1399 & $-$0.0051 & $-$0.0065 & $-$0.0813 \\
 & $10^5$ & $-$0.0115 & $-$0.0096 & $-$0.0062 & $-$0.0494 & $-$0.0471 & $-$0.0618 & 0.0025 & 0.0006 & $-$0.0189 \\
 & $10^6$ & $-$0.0130 & $-$0.0122 & $-$0.0107 & $-$0.0217 & $-$0.0183 & $-$0.0147 & 0.0028 & $-$0.0001 & $-$0.0343 \\[2pt]
(b) 20\% earlier & $10^3$ & $-$0.0293 & 0.0313 & 0.1508 & $-$0.5477 & $-$0.3873 & $-$0.3611 & $-$3.1205 & 0.1897 & $-$0.1052 \\
 & $10^4$ & $-$0.0998 & $-$0.0804 & $-$0.0429 & $-$0.2187 & $-$0.1758 & $-$0.1267 & $-$3.0901 & $-$0.0044 & $-$0.0663 \\
 & $10^5$ & $-$0.0373 & $-$0.0243 & $-$0.0027 & $-$0.0513 & $-$0.0354 & $-$0.0404 & $-$3.0968 & $-$0.0332 & 0.0011 \\
 & $10^6$ & $-$0.0175 & $-$0.0159 & $-$0.0101 & $-$0.0257 & $-$0.0230 & $-$0.0189 & $-$3.0956 & $-$0.0109 & $-$0.0301 \\[2pt]
(c) 30\% earlier & $10^3$ & $-$0.0299 & 0.0374 & 0.1458 & $-$0.6686 & $-$0.4828 & $-$0.3930 & $-$5.1728 & 0.2727 & $-$0.2057 \\
 & $10^4$ & $-$0.0700 & $-$0.0390 & 0.0257 & $-$0.1658 & $-$0.0999 & $-$0.0562 & $-$5.1736 & $-$0.0095 & 0.0022 \\
 & $10^5$ & $-$0.0502 & $-$0.0411 & $-$0.0282 & $-$0.0763 & $-$0.0839 & $-$0.0745 & $-$5.1765 & $-$0.0597 & $-$0.0294 \\
 & $10^6$ & $-$0.0189 & $-$0.0166 & $-$0.0110 & $-$0.0238 & $-$0.0174 & $-$0.0089 & $-$5.1748 & $-$0.0221 & $-$0.0244 \\[2pt]
(d) 40\% earlier & $10^3$ & $-$0.3460 & 0.0037 & 0.1347 & $-$3.9016 & $-$1.3594 & $-$1.1522 & $-$7.7035 & 0.3449 & $-$1.1842 \\
 & $10^4$ & $-$0.0935 & $-$0.0709 & $-$0.0308 & $-$0.2014 & $-$0.1687 & $-$0.1807 & $-$7.7196 & $-$0.0051 & $-$0.0941 \\
 & $10^5$ & $-$0.0560 & $-$0.0492 & $-$0.0387 & $-$0.0924 & $-$0.0901 & $-$0.1026 & $-$7.7227 & $-$0.0724 & $-$0.0334 \\
 & $10^6$ & $-$0.0100 & $-$0.0042 & 0.0053 & $-$0.0115 & $-$0.0024 & 0.0095 & $-$7.7208 & $-$0.0545 & $-$0.0217 \\
\hline
\end{tabular}
\end{table}

\begin{table}[htb!]
\centering
\caption{Simulation study: standard deviation (years) of the estimated modal age across replications. Starting age 40, 500 replications, true mode 90.}
\label{tab:app_sd}
\scriptsize
\setlength{\tabcolsep}{4pt}
\begin{tabular}{llrrrrrrrrr}
\hline
 & & \multicolumn{3}{c}{RW2 posterior} & \multicolumn{3}{c}{Raw posterior} & & & \\
Scenario & $N$ & Mean & Median & Mode & Mean & Median & Mode & Gompertz & P-spline & DDT \\
\hline
(a) unimodal & $10^3$ & 0.2834 & 0.2703 & 0.2654 & 0.9491 & 1.3924 & 1.9961 & 0.2471 & 0.2601 & 1.9076 \\
 & $10^4$ & 0.0998 & 0.0949 & 0.0911 & 0.5872 & 0.8496 & 1.1721 & 0.0820 & 0.0891 & 0.9875 \\
 & $10^5$ & 0.0317 & 0.0298 & 0.0280 & 0.3530 & 0.4657 & 0.6000 & 0.0232 & 0.0258 & 0.4061 \\
 & $10^6$ & 0.0131 & 0.0126 & 0.0124 & 0.1820 & 0.2078 & 0.2625 & 0.0071 & 0.0095 & 0.1436 \\[2pt]
(b) 20\% earlier & $10^3$ & 0.6359 & 0.6610 & 0.7460 & 0.9561 & 1.4220 & 2.0512 & 0.3304 & 0.4957 & 1.9893 \\
 & $10^4$ & 0.3885 & 0.4133 & 0.5195 & 0.6269 & 0.8742 & 1.2173 & 0.1139 & 0.2524 & 1.0709 \\
 & $10^5$ & 0.2241 & 0.2416 & 0.3218 & 0.3878 & 0.5121 & 0.6675 & 0.0362 & 0.1035 & 0.4495 \\
 & $10^6$ & 0.1206 & 0.1381 & 0.1891 & 0.1793 & 0.2079 & 0.2692 & 0.0117 & 0.0376 & 0.1489 \\[2pt]
(c) 30\% earlier & $10^3$ & 0.7440 & 0.7758 & 0.8789 & 1.1161 & 1.5337 & 2.1615 & 0.4287 & 0.6081 & 2.0808 \\
 & $10^4$ & 0.4354 & 0.4730 & 0.6043 & 0.6734 & 0.9405 & 1.2985 & 0.1238 & 0.2631 & 1.1191 \\
 & $10^5$ & 0.2466 & 0.2713 & 0.3729 & 0.3841 & 0.5262 & 0.6990 & 0.0423 & 0.1262 & 0.4747 \\
 & $10^6$ & 0.1531 & 0.1733 & 0.2259 & 0.2111 & 0.2403 & 0.2978 & 0.0147 & 0.0430 & 0.1776 \\[2pt]
(d) 40\% earlier & $10^3$ & 1.1451 & 0.8902 & 1.0500 & 2.8206 & 2.7760 & 4.1593 & 0.5202 & 0.6855 & 5.0485 \\
 & $10^4$ & 0.4954 & 0.5559 & 0.7426 & 0.7186 & 1.0339 & 1.3637 & 0.1689 & 0.3243 & 1.2007 \\
 & $10^5$ & 0.2815 & 0.3149 & 0.4482 & 0.4251 & 0.5907 & 0.7567 & 0.0510 & 0.1366 & 0.5526 \\
 & $10^6$ & 0.1504 & 0.1730 & 0.2316 & 0.2035 & 0.2370 & 0.3008 & 0.0170 & 0.0467 & 0.1759 \\
\hline
\end{tabular}
\end{table}

\begin{table}[htb!]
\centering
\caption{Simulation study: coverage and mean width (years) of the 95\% intervals for the modal age: credible intervals of the raw and RW2 posteriors (posterior of $\tilde{M}$) and bootstrap percentile intervals of the Gompertz, P-spline and DDT estimators. Starting age 40, 500 replications, true mode 90.}
\label{tab:app_coverage}
\footnotesize
\setlength{\tabcolsep}{4pt}
\begin{tabular}{ll*{10}{r}}
\hline
 & & \multicolumn{5}{c}{Coverage} & \multicolumn{5}{c}{Width} \\
Scenario & $N$ & Raw & RW2 & Gomp. & P-spl. & DDT & Raw & RW2 & Gomp. & P-spl. & DDT \\
\hline
(a) unimodal & $10^3$ & 0.9940 & 0.9860 & 0.9300 & 0.9280 & 0.9900 & 7.5840 & 1.5317 & 0.9174 & 0.9275 & 7.0263 \\
 & $10^4$ & 0.9880 & 0.9860 & 0.9040 & 0.8920 & 0.9760 & 3.9071 & 0.5177 & 0.2889 & 0.2941 & 3.4997 \\
 & $10^5$ & 0.9840 & 0.9980 & 0.9400 & 0.9340 & 0.9620 & 1.9354 & 0.1917 & 0.0917 & 0.0937 & 1.4913 \\
 & $10^6$ & 0.9460 & 0.9820 & 0.9380 & 0.9160 & 0.9420 & 0.7180 & 0.0795 & 0.0291 & 0.0307 & 0.5649 \\[2pt]
(b) 20\% earlier & $10^3$ & 0.9960 & 0.9980 & 0.0000 & 0.9300 & 0.9940 & 8.0973 & 4.2211 & 1.3758 & 1.9391 & 7.5429 \\
 & $10^4$ & 0.9820 & 1.0000 & 0.0000 & 0.9340 & 0.9700 & 4.1291 & 2.5832 & 0.4316 & 0.9343 & 3.7257 \\
 & $10^5$ & 0.9800 & 0.9980 & 0.0000 & 0.9160 & 0.9560 & 2.1268 & 1.4527 & 0.1365 & 0.3995 & 1.6955 \\
 & $10^6$ & 0.9640 & 0.9980 & 0.0000 & 0.9240 & 0.9600 & 0.7918 & 0.6327 & 0.0434 & 0.1438 & 0.6178 \\[2pt]
(c) 30\% earlier & $10^3$ & 0.9900 & 0.9920 & 0.0000 & 0.8720 & 0.9840 & 9.0573 & 4.6652 & 1.6504 & 2.1307 & 8.0380 \\
 & $10^4$ & 0.9880 & 1.0000 & 0.0000 & 0.9460 & 0.9740 & 4.2937 & 2.8579 & 0.5215 & 1.0526 & 3.9101 \\
 & $10^5$ & 0.9820 & 0.9980 & 0.0000 & 0.9020 & 0.9640 & 2.1876 & 1.6395 & 0.1647 & 0.4677 & 1.7615 \\
 & $10^6$ & 0.9420 & 0.9940 & 0.0000 & 0.8900 & 0.9420 & 0.8347 & 0.6911 & 0.0522 & 0.1639 & 0.6586 \\[2pt]
(d) 40\% earlier & $10^3$ & 0.9920 & 1.0000 & 0.0000 & 0.8640 & 0.9860 & 26.7608 & 7.5979 & 1.9931 & 3.0086 & 23.2659 \\
 & $10^4$ & 0.9800 & 0.9960 & 0.0000 & 0.9340 & 0.9800 & 4.5709 & 3.1150 & 0.6284 & 1.2145 & 4.1567 \\
 & $10^5$ & 0.9800 & 0.9960 & 0.0000 & 0.9040 & 0.9680 & 2.2963 & 1.7993 & 0.1988 & 0.5348 & 1.8921 \\
 & $10^6$ & 0.9720 & 0.9860 & 0.0000 & 0.7840 & 0.9540 & 0.9247 & 0.7527 & 0.0627 & 0.1861 & 0.7169 \\
\hline
\end{tabular}
\end{table}

\section{Sensitivity to the starting age by country}
\label{app:start_age}

Figure~\ref{fig:app_start_all} extends Figure~\ref{fig:emp_start_age} to all four countries. It shows the estimates of each method by year at starting ages 40, 50, 60 and 70 and complements the summaries in Table~\ref{tab:start_age}.

\begin{figure*}[htb!]
    \centering
    \includegraphics[width=\linewidth]{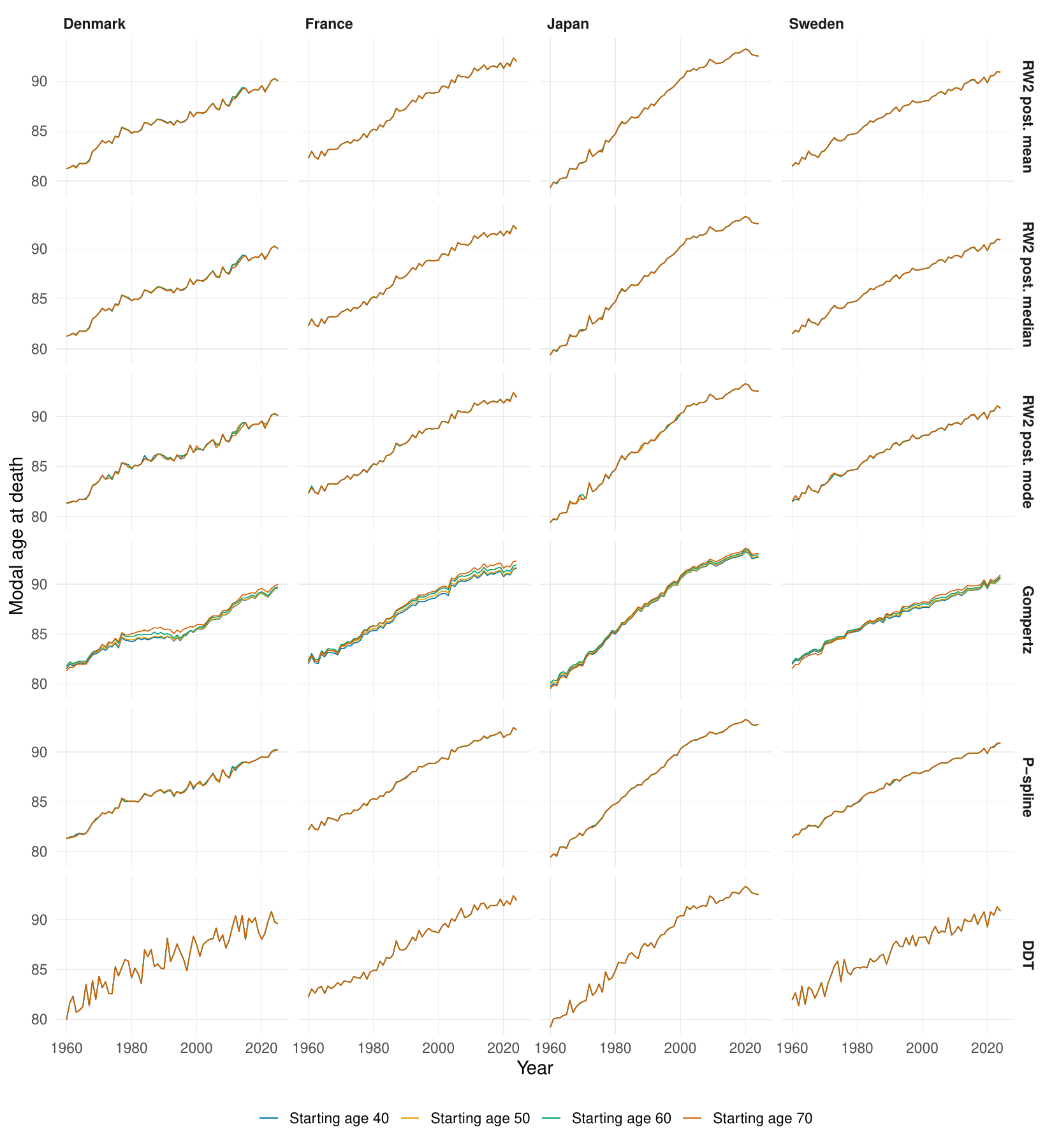}
    \caption{Modal age at death of women by year, method (rows) and country (columns), at starting ages 40, 50, 60 and 70.}
    \label{fig:app_start_all}
\end{figure*}

\section{Implementation details}
\label{app:implementation}

All computations are done in R. The code is available at the address given in the Data and materials availability statement.

\paragraph{Life table.} All methods use the constant-force life table \eqref{eq:lifetable} on single years of age. The open interval $110+$ has $q_x = 1$ and is never eligible as the mode. If the exposure is zero at some age above $x_0$, the schedule is truncated before that age, and the previous age becomes the open interval.

\paragraph{Posteriors.} The raw posterior uses $\alpha = 1/2$ and $\beta = 0$. Both posteriors use $S = 5{,}000$ draws in the simulation study and $S = 10^4$ in the application. For the RW2 posterior, the grid of $\log_{10} \lambda$ runs from $-1$ to $8$ in steps of $0.25$; the search along the grid stops once the log marginal likelihood \eqref{eq:marglik} falls 12 units below its largest value, where the weights are below $e^{-12}$. The $S$ draws are allocated to the grid values by multinomial sampling with the normalized weights. When the most probable age is the first or the last closed age, $\tilde{M}^{(s)} = x + 1/2$. The posterior mode of $\tilde{M}$ is the maximum of a Gaussian kernel density estimate of the draws with the default bandwidth of the R function \texttt{density}. The banded computations of Appendix~\ref{app:computation} are coded in C++ through Rcpp. In the application (261 population-years, starting age 40), the posterior weight at either end of the grid was below $10^{-63}$, grid values with weight above $0.001$ lay between $\log_{10}\lambda = 1.5$ and $6.25$, and the posterior means of $\log_{10}\lambda$ ranged from $1.9$ to $4.9$. In the simulation, the posterior means ranged from $2.5$ to $7.6$; values near the upper end occur in the unimodal scenario, where the true log rates lie on a straight line.

\paragraph{Interpolation without noise.} To measure the error of the interpolation \eqref{eq:interp} itself, we applied the life table \eqref{eq:lifetable} and \eqref{eq:interp} to the exact central death rates of single Gompertz distributions with $b = 0.14$ and modes at $90$, $90.25$, $90.5$ and $90.75$ (starting age 40). The interpolated modes are $89.987$, $90.227$, $90.471$ and $90.724$, errors of $-0.013$ to $-0.029$ years. Using the exact interval probabilities of the density instead of the constant-force life table changes these values by less than $0.001$, so the error comes from the interpolation, not from the life table.

\paragraph{Gompertz.} The model $\log m_x = \log a + b(x + 1/2)$ is fitted by Poisson maximum likelihood, with offset $\log E_x$, to the ages from $x_0$ to $\omega - 1$. The mode is $\log(b/a)/b$.

\paragraph{P-spline.} The R package MortalitySmooth \citep{camarda2012mortalitysmooth} is no longer distributed through CRAN, so we coded the same model. The log rate at $x + 1/2$ is a combination of cubic B-splines on equally spaced knots, with $\lfloor (\omega - x_0)/5 \rfloor$ intervals (at least four), and a second-order difference penalty on the coefficients. The model is fitted to the counts by penalized iteratively reweighted least squares with offset $\log E_x$, and the smoothing parameter minimizes the Bayesian information criterion over $\log_{10} \lambda \in \{-2, -1.5, \dots, 7\}$. The mode is the maximum of the fitted density $f(t) = \mu(t) \exp\{-\int_{x_0}^{t} \mu(u)\, du\}$ on a grid of $0.01$ years, refined by the vertex of the parabola through the three grid points around the maximum.

\paragraph{DDT.} We compute the four conditions of \citet[Equation 9]{vazquez2024longevity} from the observed rates and the corresponding life-table deaths, and locate each change of sign from positive to negative by linear interpolation between integer ages. Every zero of the first-difference condition at which the second-difference condition is negative is a candidate local maximum. The zeros of the two rate conditions closest to the candidate, within two years, are averaged with it; a condition with no zero within two years is left out. Among the candidates, the estimate is the one with the largest life-table deaths. The original description does not specify how zeros are matched when the conditions cross zero more than once, so this matching rule is our choice. Finally, we add $1/2$ to the estimate. The published conditions treat $d_x$ as the density of deaths at exact age $x$, but $d_x$ counts the deaths in $[x, x+1)$, whose center is $x + 1/2$. The Gompertz and P-spline fits evaluate the hazard at $x + 1/2$, and the interpolation \eqref{eq:interp} locates the mode within $[x, x+1)$. Without the shift, DDT estimates were on average about half a year below the true mode in the simulation study and below the other estimates in the application.

\paragraph{Bootstrap intervals.} For the Gompertz, P-spline and DDT estimators, each simulated data set is resampled 200 times as $D_x^\ast \sim \text{Poisson}(D_x)$ with the exposures held fixed, the mode is re-estimated in each replicate, and the 2.5\% and 97.5\% quantiles of the replicates form the 95\% interval. With 200 replicates, each tail quantile rests on about five draws. For the P-spline, the smoothing parameter is fixed at the value selected on the original data. To assess the effect of fixing the smoothing parameter, we repeated the P-spline bootstrap with 500 replicates, first with fixed smoothing and then with the parameter re-selected by the Bayesian information criterion in every replicate. This check used the first 200 data sets of scenarios (a) and (d) at $N = 10^3$ and $10^6$. With the smoothing parameter fixed, coverage was $0.94$ and $0.91$ in scenario (a) at $N = 10^3$ and $10^6$, and $0.89$ and $0.78$ in scenario (d). With re-selection, coverage was $0.96$, $0.97$, $0.94$ and $0.80$. The mean width rose from $0.94$ to $1.84$ years and from $0.03$ to $0.06$ years in scenario (a), and from $3.27$ to $5.28$ years in scenario (d) at $N = 10^3$. It stayed at $0.19$ years in scenario (d) at $N = 10^6$. The Monte Carlo standard errors of these coverages are $0.01$ to $0.03$.


\end{document}